\documentclass[12pt,draftcls,onecolumn]{IEEEtran}

\usepackage{graphicx}
\usepackage{array}
\usepackage{mathrsfs}
\usepackage{amsfonts}
\usepackage{setspace}
\usepackage{amsmath}
\usepackage{amssymb}
\usepackage{mathrsfs}
\input xy
\xyoption{all}

\newtheorem{theorem}{Theorem}
\newtheorem{lemma}{Lemma}

\newtheorem{definition}{Definition}
\newtheorem{corollary}{Corollary}

\newtheorem{remark}{Remark}
\newtheorem{example}{Example}
\newtheorem{problem}{Problem}

\begin{document}
%

\title{Technical Report: NUS-ACT-12-001-Ver.3:
Cooperative Tasking for Deterministic Specification
Automata\\Preprint, Submitted for publication}

\author{Mohammad~Karimadini,
        and~Hai~Lin,
\thanks{M. Karimadini is with the Department of Electrical and Computer Engineering, National University of Singapore,
Singapore. H. Lin is from the Department of Electrical Engineering
University of Notre Dame, IN 46556 USA. Corresponding author, H. Lin
\{\small hlin1@nd.edu\}. } }

\maketitle \thispagestyle{empty} \pagestyle{empty}

\begin{abstract}
In our previous work \cite{Automatica2010-2-agents-decomposability},
a divide-and-conquer approach was proposed for cooperative tasking
 among multi-agent systems. The basic idea is to decompose a requested
global specification into subtasks for individual agents such that
the fulfillment of these subtasks by each individual agent leads to
the satisfaction of the global specification as a team. It was shown
that not all tasks can be decomposed. Furthermore, a necessary and
sufficient condition was proposed for the decomposability of a task
automaton between two cooperative agents. The current paper
continues the results in \cite{Automatica2010-2-agents-decomposability} and proposes
necessary and sufficient conditions for task decomposability with
respect to arbitrary finite number of agents. It is further shown
that the fulfillment of local specifications can guarantee the
satisfaction of the global specification.
This work
provides hints for the designers on how to rule out the
indecomposable task automata and enforce the decomposability
conditions. The result therefore may pave the way towards a new
perspective for decentralized cooperative control of multi-agent
systems.
\end{abstract}

\section{INTRODUCTION}
Multi-agent system emerges as a rapidly
developing multi-disciplinary area with
promising applications in
assembling and transportation, parallel computing, distributed planning and scheduling, rapid emergency
response and multi-robot systems \cite{Lima2005,
 Bo2010, Daneshfar2009, bullo2009distributed}. The significance of
multi-agent systems roots in the power of parallelism and
cooperation between simple components that synergically lead to
sophisticated capabilities, robustness and functionalities \cite{Lima2005, Lesser1999}.
The cooperative control of distributed multi-agent systems,
however, is still in its infancy with significant practical and
theoretical challenges that are difficult to be formulated and
tackled by the traditional methods \cite{Tabuada2006, Kloetzer2007}.
Among these challenges, one essential issue is the top-down cooperative control to achieve a desired global behavior through the design of local control laws and interaction rules \cite{Crespi2008}. Top-down cooperative control is typically synthesized in two levels of abstraction: control level and planning (supervisory) layer \cite{murray2007recent}.

 Control level deals with the time-driven continues dynamics of each agent, dynamic topology and on-line interactions among the agents, in order for real-time tracking of exact trajectories, collision avoidance, formation stability and optimal performance \cite{bullo2009distributed, martinez2007motion}. For this purpose, several innovative approaches have been developed such as
biomimicry of biological swarms and symbolic swarming \cite{Reynolds1987, Kloetzer2007},
consensus seeking and formation
stabilization  \cite{Jadbabaie2003}, navigation functions for distributed
formation \cite{Tanner2005}, artificial potential functions \cite{Reif1995}, graph Laplacians for the associated neighborhood graphs
\cite{Jadbabaie2003, Ji2009, tanner2007flocking}, graph-based formation stabilization and coordination \cite{zelazo2011graph, fax2004information, tanner2004leader, cortes2009global},
passivity-based control \cite{arcak2007passivity, zelazo2011edge},
distributed predictive control \cite{richards2004decentralized}, game theory-based coordinations
\cite{semsar2009multi} and potential games and mechanism design \cite{heydenreich2007games, christodoulou2004coordination}.
    These methods successfully model the interactions among the agents using the topology graph and apply Lyaponuv-like energy functions and optimization methods for stabilization and formation of the continuous states of the agents.

In the planning level, on the other hand, one concerns with the event-driven dynamics, discrete modes and logical specifications such as visiting successive regions, orchestrating between local controllers and path planning for the control layer. One of the main challenges in the planning level is the cooperative tasking to allocate local tasks to each agent such that a desired logical specification is globally satisfied by the team.
Confining to the planning level, this paper and its companion
\cite{Automatica2010-2-agents-decomposability} aim at developing a
top-down correct-by-design method for distributed coordination and
cooperative tasking of multi-agent systems such that the group of agents, as a
team, can achieve the specified logical requirements, collectively.
We assume here that
the global specification is given as a finite deterministic
automaton that is simpler to be characterized; covers a wide
class of tasks
in the context of supervisory control of discrete event
systems \cite{Ramadge1987}, and can uniquely encode the sequence of events in a finite memory space
using the notions of states and transition relations. Accordingly, the logical behavior of a
multi-agent system can be modeled as a parallel distributed
system \cite{Mukund2002} that having the union of local event sets, allows the agents to individually transit on their private events, while synchronize on shared events for cooperative tasks.
Since in this set up, each agent will have access to its local set of sensor readings and
actuator commands, the interpretation of each agent from the global
task automaton can be obtained through natural projection of the
global task into the corresponding local observable events \cite{Morin1998}.
The composition of these local task should be able to retrieve the global task in order to perform the cooperative tasking.
For this purpose, we are particularly interested in task automaton decomposability (also called synthesis modulo problem)
to understand that under what conditions the collective perception of the team
from the global specification (the parallel composition of local task automata) is equivalent to the original global
task. Generally, three types of equivalence relations have been studied in
the top-down cooperative control in order to compare the global task
automaton with the collective one \cite{Mukund2002, Morin1998, Chen2010}:
isomorphism, language equivalence and bisimulation. Bisimulation-based decomposability is less restrictive than synthesis modulo isomorphism and more applicable in control applications, while it is more expressive than language separability \cite{Willner1991}. Moreover, it preserves the nondeterminism that might appear in  the collective tasks, even for deterministic global task automata.

Given a task automaton and the distribution of its events among the
agents, we have shown in
\cite{Automatica2010-2-agents-decomposability} that it is not always
possible to decompose an automaton into sub-automata by natural
projections, where the parallel composition of these sub-automata is
bisimilar to the original automaton, and subsequently necessary and sufficient conditions were identified for the decomposability of deterministic task automaton with respect to two local event sets.
For more than two agents, a sufficient condition was proposed in
\cite{Automatica2010-2-agents-decomposability} by introducing a
hierarchical approach to iteratively use the decomposability for two
agents. Therefore, the main part of this paper is set to provide new
necessary and sufficient conditions for the decomposability of a
task automaton with respect to an arbitrary finite number of agents.
The extension is not straightforward and requires logical modifications
on the conditions for the two-agent result.

Please note that the main contribution of the current work is to
gain insights on decomposability of a task automaton rather than
checking the decomposability itself. The proposed decomposability
conditions provide us with hints on how to rule out the
indecomposable automata and how the configuration of local
transitions and distribution of events among the agents should be in
order for decomposability. It is shown that an automaton is
decomposable if and only if any
decision on the order or selection between two transitions can be  made by at least one of the agents, the interleaving of any pair of strings
after synchronizing on a shared event does not introduce a new string
that is not in the original automaton (the interleaving of local
task automata does not allow an illegal global behavior), and each
local task automaton bisimulates a deterministic automaton (to
ensure that the collection of local tasks does not disallow a legal
global behavior). These insights are important since they
give us guidelines on how to set a global task to be fulfilled,
cooperatively, by the team of agents.

The rest of the paper is organized as follows. Preliminary notations, definitions and problem formulation are represented in Section \ref{PROBLEM FORMULATION}. Section \ref{TASK DECOMPOSITION
FOR $n$ AGENTS} introduces the necessary and sufficient conditions
for decomposability of an automaton with respect to parallel
composition and an arbitrary finite number of local event sets.
Finally, the paper concludes with remarks and discussions
in Section \ref{CONCLUSION}. The proofs of lemmas are provided in the Appendix.
\section{PROBLEM FORMULATION}\label{PROBLEM FORMULATION}
We first recall the definition of deterministic automaton \cite{Kumar1999}.
\begin{definition}(Automaton)\label{Automaton}
A deterministic automaton is a tuple $A := \left(Q, q_0, E, \delta
\right)$ consisting of a set of states $Q$; an initial state $q_0\in
Q$; a set of events $E$ that causes transitions between the states,
and a transition relation $\delta \subseteq Q \times E\times Q$
(with a partial map $\delta: Q \times E \to Q$), such that $(q, e,
q^{\prime})\in \delta$ if and only if state $q$ is transited to
state $q^{\prime}$ by event $e$, denoted by
$q\overset{e}{\underset{}\rightarrow } q^{\prime}$ (or $\delta(q, e)
= q^{\prime}$).
In general the automaton also has an argument $Q_m
\subseteq Q$ of marked (accepting or final) states to assign a
meaning of accomplishment to some states. For an automaton whose
each state represents an accomplishment of a stage of the
specification, all states can be considered as marked states and
$Q_m$ is omitted from the tuple.
\end{definition}

With an abuse of notation, the definitions of the transition
relation can be extended from the domain of $Q \times E$
into the
domain of $Q \times E^*$
 to define transitions over strings $s\in
E^*$, where $E^*$ stands for the $Kleene-Closure$ of $E$ (the
collection of all finite sequences of events over elements of $E$).
\begin{definition}(Transition on strings)\label{Transition on strings}
For a deterministic automaton the existence of a transition over a
string $s\in E^*$ from a state $q\in Q$ is denoted by $\delta(q,
s)!$ and inductively defined as $\delta(q,\varepsilon) = q$, and
$\delta(q,se)=\delta(\delta(q,s),e)$ for $s\in E^*$ and $e\in E$.
The existence of a set $L\subseteq E^*$ of strings from a state
$q\in Q$ is then denoted as $\delta(q, L)!$ and read as $\forall
s\in L: \delta(q, s)!$.
\end{definition}

The transition relation is a partial relation, and in general some
of the states might not be accessible from the initial state.
\begin{definition}\label{accessible}
The operator $Ac(.)$ \cite{Cassandras2008} is then defined by
excluding the states and their attached transitions that are not
reachable from the initial state as $Ac(A) = \left(Q_{ac}, q_0, E,
\delta_{ac} \right)$ with
$Q_{ac}=\{q\in Q|\exists s\in E^*, q \in
\delta (q_0, s)\}$ and $\delta_{ac}= \{(q, e, q^{\prime})\in \delta|e\in E, q, q^{\prime}\in Q_{ac}\}$.
Since $Ac(.)$ has  no effect on the behavior of the automaton, from
now on we take $A = Ac(A)$.
\end{definition}



The qualitative behavior of a deterministic system is described by its language defined as
\begin{definition}(Language, language equivalent automata)
For a given automaton $A$, the language generated by $A$ is defined
as $L(A):= \{s\in E^*|\delta(q_0, s)!\}$. Two automata $A_1$ and
$A_2$ are said to be language equivalent if $L(A_1) = L(A_2)$.
\end{definition}

In cooperative tasking, each agent has a local observation from the
global task: the perceived global task, filtered by its local event
set, i.e., through a mapping over each agent's
event set, as the interpretation of each agent from the global
task.
Particularly,  natural projections
$P_{E_i}(A_S)$ are obtained from $A_S$ by replacing its events that
belong to $E\backslash E_i$ by $\varepsilon$-moves, and then,
merging the $\varepsilon$-related states. The  $\varepsilon$-related
states form equivalent classes defined as follows.
\begin{definition}(Equivalent class of states, \cite{Morin1998})\label{Equivalent class of states}
Consider an automaton $A_S=(Q, q_0, E, \delta)$ and an event set
$E^{\prime}\subseteq E$. Then, the relation $\sim_{E^{\prime}}$ is
the minimal equivalence relation on the set $Q$ of states such that
$q^{\prime}\in\delta(q, e) \wedge e\notin E^{\prime}\Rightarrow
q\sim_{E^{\prime}} q^{\prime}$, and $[q]_{E^{\prime}}$ denotes the
equivalence class of $q$ defined on $\sim_{E^{\prime}}$.
 The set of equivalent classes of states
over $\sim_{E^{\prime}}$, is denoted by $Q_{/\sim_{E^{\prime}}}$ and defined as
$Q_{/\sim_{E^{\prime}}} = \{[q]_{E^{\prime}}|q\in Q\}$.

$\sim_{E^{\prime}}$ is an equivalence relation as it is reflective ($q\sim_{E^{\prime}} q$), symmetric ($q\sim_{E^{\prime}}
q^{\prime} \Leftrightarrow q^{\prime}\sim_{E^{\prime}}
q$) and transitive ($q\sim_{E^{\prime}}
q^{\prime} \wedge q^{\prime}\sim_{E^{\prime}}
q^{\prime\prime} \Rightarrow q\sim_{E^{\prime}}
q^{\prime\prime}$).

It should be noted that the relation $\sim_{E^{\prime}}$ can be defined
for any $E^{\prime}\subseteq E$, for example, $\sim_{E_i}$ and $\sim_{E_i\cup E_j}$,
respectively denote the equivalence relations with respect to
 ${E_i}$ and $E_i\cup E_j$. Moreover, when it is clear from the context,
  $\sim_i$ is used to denote $\sim_{E_i}$ for simplicity.
\end{definition}

Next, natural projection over strings is denoted by $p_{E^{\prime}}
: E^*\rightarrow E^{\prime*}$, takes a string from the event set $E$
and eliminates events in it that do not belong to the event set
$E^{\prime}\subseteq E$. The natural projection is formally defined
on the strings as
\begin{definition}(Natural Projection on String,
\cite{Cassandras2008})\label{Natural Projection on String} Consider
a global event set $E$ and an event set $E^{\prime}\subseteq E$.
Then, the natural projection $p_{E^{\prime}}:E^*\rightarrow
E^{\prime *}$ is inductively defined as
 $p_{E^{\prime}}(\varepsilon)=\varepsilon$, and
 $\forall s\in E^*, e\in E: p_{E^{\prime}}(se)=\left\{
  \begin{array}{ll}
  p_{E^{\prime}}(s)e & \hbox{if $e\in E^{\prime}$;} \\
  p_{E^{\prime}}(s) & \hbox{otherwise.}
  \end{array}
\right.$
\end{definition}

The natural projection is then formally defined on an automaton as
follows.
\begin{definition}(Natural Projection on Automaton)\label{Natural Projection on Automaton}
Consider a deterministic automaton $A_S = (Q, q_0, E, \delta)$ and
an event set $E^{\prime}\subseteq E$. Then,
$P_{E^{\prime}}(A_S)=(Q_i = Q_{/\sim_{E^{\prime}}},
[q_0]_{E^{\prime}}, E^{\prime}, \delta^{\prime})$, with
$[q^{\prime}]_{E^{\prime}}\in\delta^{\prime}([q]_{E^{\prime}}, e)$
if there exist states $q_1$ and $q_1^{\prime}$ such that
$q_1\sim_{E^{\prime}} q$, $q_1^{\prime}\sim_{E^{\prime}}
q^{\prime}$, and $\delta(q_1, e) = q^{\prime}_1$. Again,
$P_{E^{\prime}}(A_S)$ can be defined into any event set
$E^{\prime}\subseteq E$. For example, $P_{E_i}(A_S)$ and $P_{E_i\cup
E_j}(A_S)$, respectively denote the natural projections of $A_S$
into $E_i$ and $E_i\cup E_j$. When it is clear from the context,
$P_{E_i}$ is replaced with $P_i$, for simplicity.
\end{definition}

The collective task is then obtained using the parallel composition of local task automata, as the perception of the team from the global task.
\begin{definition} (Parallel Composition \cite{Kumar1999}) \label{parallel
composition}\\ Let $A_i=\left( Q_i,q_i^0,E_i,\delta _i\right)$,
$i=1,2$, be automata. The parallel composition (synchronous
composition) of $A_1$ and $A_2$ is the automaton $A_1||A_2=\left(Q =
Q_1 \times Q_2, q_0 = (q_1^0, q_2^0), E = E_1 \cup E_2,
\delta\right)$, with $\delta$ defined as
 $\forall (q_1, q_2)\in Q, e\in E:\\\delta(\left(q_1, q_2),
   e\right)=
    \left\{
\begin{array}{ll}
    \left(\delta_1(q_1, e), \delta_2(q_2, e)\right), & \hbox{if $\delta _1(q_1,e)!, \delta _2(q_2,e)!, e\in E_1 \cap E_2$;}\\
    \left(\delta_1(q_1, e), q_2\right), & \hbox{if $\delta _1(q_1,e)!, e\in E_1 \backslash E_2$;} \\
    \left(q_1, \delta_2(q_2, e)\right), & \hbox{if $\delta _2(q_2,e)!, e\in E_2 \backslash E_1$;} \\
     \hbox{undefined}, & \hbox{otherwise.}
\end{array}\right.$

The parallel composition of $A_i$, $i=1,2,...,n$ is called parallel
distributed system, and is defined based on the associativity
property of parallel composition \cite{Cassandras2008} as
$\overset{n}{\underset{i=1}{\parallel\ } }A_i := A_1\parallel\
...\parallel\   A_n := A_n\parallel \left(A_{n-1}
\parallel \left( \cdots \parallel \left( A_2\parallel
A_1 \right)\right)\right)$.
\end{definition}

The obtained collective task is then compared with the original global task automaton using the bisimulation relation, in order to ensure that the team of agents understands the global specification, collectively.
\begin{definition}(Bisimulation \cite{Cassandras2008})
\label{simulation}
Consider two automata $A_i=( Q_i, q_i^0$, $E,
\delta _i)$, $i=1, 2$.
The automaton $A_1$ is said to be similar to $A_2$ (or $A_2$ simulates $A_1$), denoted
by $A_1\prec A_2$, if there exists a simulation relation from $A_1$ to $A_2$ over $Q_1$, $Q_2$ and with respect to $E$, i.e., (1) $(q_1^0, q_2^0) \in R$, and (2) $\forall\left( {q_1 ,q_2 } \right) \in R,
q'_1 \in \delta_1(q_1, e)$, then $\exists q_2^{\prime}\in Q_2$ such
that $q'_2 \in \delta_2(q_2, e)$, $\left( {q'_1 ,q'_2 } \right) \in
R$ \cite{Cassandras2008}.

Automata $A_1$ and $A_2$ are said to be bisimilar (bisimulate each
other), denoted by $A_1\cong A_2$ if $A_1\prec A_2$ with a simulation relation $R_1$, $A_2\prec A_1$
with a simulation relation $R_2$ and $R_1^{-1} = R_2$ \cite{liu2011bisimilarity}, where $R_1^{-1}= \{(y,x)\in Q_2\times Q_1|(x,y)\in R_1\}$.
\end{definition}

Based on these definitions we may now formally define the
decomposability of an automaton with respect to parallel composition
and natural projections as follows.
\begin{definition}(Automaton decomposability)
 A task automaton $A_S$ with the event set $E$ and local event sets
 $E_i$, $i=1,..., n$, $E = \overset{n}{\underset{i=1}{\cup} } E_i$, is
said to be decomposable with respect to parallel composition and
natural projections $P_i$,
$i=1,\cdots, n$, when $\overset{n}{\underset{i=1}{\parallel} } P_i
\left( A_S \right) \cong A_S$.
\end{definition}


\begin{example}\label{hierarchical algorithm is sufficient} The
automaton $A_S$:
  \xymatrix@R=0.5cm{
&   \bullet\ar[r]^{e_2}& \bullet \ar[r]^{b}& \bullet \ar[r]^{e_3}& \bullet   \ar[r]^{c}& \bullet \ar[r]^{e_5}& \bullet \ar[r]^{d}& \bullet  & \\
\ar[r]&  \bullet \ar[u]^{a} \ar[d]_{e_1} & &  &  &&&\bullet             \\
 &   \bullet\ar[r]^{a}& \bullet \ar[r]^{e_2}& \bullet \ar[r]^{b}& \bullet \ar[r]^{e_3}&
 \bullet  \ar[r]^{c}& \bullet \ar[r]^{e_5}& \bullet \ar[u]^{d}&
                }\\
                 with $E=E_1\cup E_2\cup E_3$, $E_1=\{a, c, d, e_1, e_5\}$,
                $E_2=\{a, b, d,
                e_2\}$, $E_3=\{b, c, e_3\}$, $P_1(A_S)$:
 \xymatrix@R=0.1cm{
&  & \bullet\ar[r]^{c}& \bullet \ar[r]^{e_5}& \bullet \ar[r]^{d}& \bullet  & \\
\ar[r]&  \bullet \ar[ur]^{a} \ar[dr]_{e_1} & &  &               \\
 &  & \bullet\ar[r]^{a}& \bullet \ar[r]^{c}& \bullet \ar[r]^{e_5}& \bullet \ar[r]^{d}& \bullet
                } ,
             $P_2(A_S)\cong \xymatrix@C=0.5cm{
     \ar[r]&  \bullet \ar[r]^{a} &  \bullet \ar[r]^{e_2} &  \bullet \ar[r]^{b} &
     \bullet \ar[r]^{d} &  \bullet
            }$ and
            $P_3(A_S)\cong \xymatrix@C=0.5cm{
     \ar[r]&  \bullet \ar[r]^{b} &  \bullet \ar[r]^{e_3} &  \bullet \ar[r]^{c} &  \bullet
            }$, is decomposable as
$A_s\cong P_1(A_S)||P_2(A_S)||P_3(A_S)$.
\end{example}


\begin{remark} \label{Bisimilar instead of equal}
Since bisimilarity is an equivalence relation it is also transitive,
and hence $P_i(A_S)$'s can be denoted as being bisimilar, rather
than equal to the drawn automata, since $P_i^{\prime}(A_S) \cong
P_i(A_S)$, $i = 1, \ldots, n$, and
$\overset{n}{\underset{i=1}{\parallel} } P_i^{\prime}(A_S) \cong
A_S$ is equivalent to $\overset{n}{\underset{i=1}{\parallel} }
P_i(A_S)\cong A_S$.

\end{remark}

In \cite{Automatica2010-2-agents-decomposability}, we proposed a necessary and sufficient condition for the task decomposability with respect to two agents. For more than two agents a hierarchical algorithm was proposed to iteratively use the decomposability for two agents. The algorithm is a sufficient condition only, as it can decompose the task automaton if at each stage the task is decomposable with respect to one local event set and the rest of agents. For instance in Example \ref{hierarchical algorithm is sufficient} $A_S$, is decomposable as $A_s\cong P_1(A_S)||P_2(A_S)||P_3(A_S)$, and choosing any of local event sets $E_1$, $E_2$ and $E_3$
the algorithm passes the first stage of hierarchical decomposition, as $A_s\cong
P_1(A_S)||(P_2(A_S)||P_3(A_S))\cong
P_3(A_S)||(P_1(A_S)||P_2(A_S))\cong P_2(A_S)||(P_1(A_S)||P_3(A_S))$,
 but it sucks at the second step, as $P_{E_2\cup E_3}(A_S)\ncong P_2(A_S)||P_3(A_S)$,
$P_{E_1\cup E_2}(A_S)\ncong P_1(A_S)||P_2(A_S)$ and $P_{E_1\cup
E_3}(A_S)\ncong P_1(A_S)||P_3(A_S)$).
Moreover, it is possible to show by counterexamples that not all automata are
decomposable with respect to parallel composition and natural
projections (see following example). Then, a natural follow-up
question is what makes an automaton decomposable.
\begin{problem}\label{solution} Given a deterministic task automaton $A_S$
with event set $E = \overset{n}{\underset{i=1}{\cup} } E_i$ and
local event sets $E_i$, $i=1,\cdots, n$, what are the necessary and
sufficient conditions that $A_S$ is decomposable with respect to
parallel composition and natural projections $P_i$, $i=1,\cdots, n$,
such that $\overset{n}{\underset{i=1}{\parallel} } P_i \left( A_S
\right) \cong A_S$?
\end{problem}

\section{TASK DECOMPOSITION FOR $n$ AGENTS}\label{TASK DECOMPOSITION FOR $n$ AGENTS}

The main result on task automaton
decomposition is given as follows.
\begin{theorem}\label{Decomposability Theorem for n agents}
A deterministic
automaton $A_S = \left( {Q, q_0 , E = \bigcup\limits_{i = 1}^n {E_i}
, \delta } \right)$ is decomposable with respect to parallel
composition and natural projections $P_i$, $i=1,...,n$ such that $A_S \cong \mathop {||}\limits_{i =
1}^n P_i \left( {A_S } \right)$ if and only if
\begin{itemize}
\item $DC1$: $\forall e_1,
e_2 \in E, q\in Q$: $[\delta(q,e_1)!\wedge \delta(q,e_2)!]\\
\Rightarrow [\exists E_i\in\{E_1, \ldots, E_n\}, \{e_1,
e_2\}\subseteq E_i]\vee[\delta(q, e_1e_2)! \wedge \delta(q,
e_2e_1)!]$;
\item $DC2$: $\forall e_1, e_2 \in E,  q\in Q$, $s\in E^*$: $[\delta(q,
e_1e_2s)!\vee \delta(q, e_2e_1s)!]\\ \Rightarrow [\exists
E_i\in\{E_1, \ldots, E_n\}, \{e_1, e_2\}\subseteq E_i]\vee [
\delta(q, e_1e_2s)!\wedge \delta(q, e_2e_1s)!]$;
\item $DC3$:
 $\forall q, q_1, q_2 \in Q$, strings $s, s^{\prime}$ over $E$, $\delta(q, s)= q_1$, $\delta(q,
s^{\prime})= q_2$, $\exists i, j \in \{1, \cdots, n\}$, $i \neq j$, $p_{E_i\cap E_j}(s)$, $p_{E_i\cap E_j}(s^{\prime})$ start with $a\in E_i\cap E_j$: $\mathop {||}\limits_{i = 1}^n P_i \left( {A} \right)\prec
A_S(q)$ (where $A:=\xymatrix@R=0.1cm{
                \ar[r]&  \bullet\ar[r]^{s}\ar[dr]_{s^{\prime}}&  \bullet \\
                &&\bullet    }$ and $A_S(q)$ denotes an automaton that is obtained from $A_S$, starting from $q$, and
\item $DC4$: $\forall i\in\{1,...,n\}$, $x, x_1, x_2 \in Q_i$, $x_1\neq x_2$,
$e\in E_i$, $t\in E_i^*$, $x_1\in\delta_i (x, e)$, $x_2\in\delta_i
(x, e)$: $\delta_i (x_1, t)! \Leftrightarrow \delta_i(x_2, t)!$.
 \end{itemize}
\end{theorem}

\begin{proof}
In order for $A_S\cong \mathop {||}\limits_{i = 1}^n P_i \left( {A_S
} \right)$, from the definition of bisimulation, it is required to
have $A_S\prec \mathop {||}\limits_{i = 1}^n P_i \left( {A_S }
\right)$; $\mathop {||}\limits_{i = 1}^n P_i \left( {A_S }
\right)\prec A_S$, and the simulation relations are inverse of each
other. These requirements are provided by the following three
lemmas.

Firstly, $\mathop {||}\limits_{i = 1}^n P_i \left( {A_S } \right)$
always simulates $A_S$. Formally:
\begin{lemma}\label{always As is similar to parallel decomposition}
Consider a deterministic automaton $A_S  = \left( {Q,q_0 ,E =
\bigcup\limits_{i = 1}^n {E_i ,\delta } } \right)$ and natural
projections $P_i$, $i=1,...,n$.
Then, it always holds that $A_S \prec \mathop {||}\limits_{i = 1}^n
P_i \left( {A_S } \right)$.
\end{lemma}

The similarity of $\mathop {||}\limits_{i = 1}^n P_i \left( {A_S }
\right)$ to $A_S$, however, is not always true
(see Example \ref{undecomposable automata}), and needs some
conditions as stated in the following lemma.

\begin{lemma}\label{Similarity of parallel decomposition to As}
Consider a deterministic automaton $A_S  = \left( {Q,q_0 ,E =
\bigcup\limits_{i = 1}^n {E_i ,\delta } } \right)$ and natural
projections $P_i$, $i=1,...,n$.
Then, $\mathop {||}\limits_{i = 1}^n P_i \left( {A_S } \right)\prec
A_S$ if and only if
\begin{itemize}
\item $DC1$: $\forall e_1,
e_2 \in E, q\in Q$: $[\delta(q,e_1)!\wedge \delta(q,e_2)!]\\
\Rightarrow [\exists E_i\in\{E_1, \ldots, E_n\}, \{e_1,
e_2\}\subseteq E_i]\vee[\delta(q, e_1e_2)! \wedge \delta(q,
e_2e_1)!]$;
\item $DC2$: $\forall e_1, e_2 \in E,  q\in Q$, $s\in E^*$: $[\delta(q,
e_1e_2s)!\vee \delta(q, e_2e_1s)!]\\ \Rightarrow [\exists
E_i\in\{E_1, \ldots, E_n\}, \{e_1, e_2\}\subseteq E_i]\vee [
\delta(q, e_1e_2s)!\wedge \delta(q, e_2e_1s)!]$;
\item $DC3$:
 $\forall q, q_1, q_2 \in Q$, strings $s, s^{\prime}$ over $E$, $\delta(q, s)= q_1$, $\delta(q,
s^{\prime})= q_2$, $\exists i, j \in \{1, \cdots, n\}$, $i \neq j$, $p_{E_i\cap E_j}(s)$, $p_{E_i\cap E_j}(s^{\prime})$ start with $a\in E_i\cap E_j$: $\mathop {||}\limits_{i = 1}^n P_i \left( {A} \right)\prec
A_S(q)$ (where $A:=\xymatrix@R=0.1cm{
                \ar[r]&  \bullet\ar[r]^{s}\ar[dr]_{s^{\prime}}&  \bullet \\
                &&\bullet    }$ and $A_S(q)$ is
                 an automaton that is obtained from $A_S$, starting from $q$).
 \end{itemize}
\end{lemma}

Next, we need to show that for two simulation relations $R_1$ (for
$A_S \prec \mathop {||}\limits_{i = 1}^n P_i \left( {A_S } \right)$)
and $R_2$ (for $\mathop {||}\limits_{i = 1}^n P_i \left( {A_S }
\right)\prec A_S$) defined by the above two lemmas, $R_1^{-1} =
R_2$.

\begin{lemma}\label{symmetric}
Consider an automaton $A_S=(Q, q_0, E=E_1\cup E_2, \delta)$ with
natural projections $P_i$, $i=1,...,n$. If $A_S$ is deterministic,
$A_S\prec \mathop {||}\limits_{i = 1}^n P_i \left( {A_S } \right)$
with the simulation relation $R_1$ and $\mathop {||}\limits_{i =
1}^n P_i \left( {A_S } \right)\prec A_S$ with the simulation
relation $R_2$, then $R^{-1}_1 = R_2$ (i.e., $\forall q\in Q$, $z\in
Z$: $(z, q)\in R_2 \Leftrightarrow (q, z)\in R_1$) if and only if
$DC4$: $\forall i\in\{1,...,n\}$, $x, x_1, x_2 \in Q_i$, $x_1\neq
x_2$, $e\in E_i$, $t\in E_i^*$, $x_1\in\delta_i (x, e)$,
$x_2\in\delta_i (x, e)$: $\delta_i (x_1, t)! \Leftrightarrow
\delta_i(x_2, t)!$.
\end{lemma}

Now, Theorem \ref{Decomposability Theorem for n agents} is proven as follows.
Firstly, conditions $DC1$ and $DC2$ in Theorem \ref{Decomposability Theorem for n agents}
are equivalent to the respective conditions in Lemma \ref{Similarity of parallel decomposition to As} due to the logical equivalences
$(p\wedge
q)\Rightarrow r \equiv q \Rightarrow (\neg p\vee r)$ and $p
 \Leftrightarrow q \equiv \left( {p \vee q} \right) \Rightarrow
\left( {p \wedge q} \right)$, for any expressions $p$ and $q$.
Then, according to Definition \ref{simulation}, $A_S\cong \mathop
{||}\limits_{i = 1}^n P_i \left( {A_S } \right)$ if and only if
$A_S\prec \mathop {||}\limits_{i = 1}^n P_i \left( {A_S } \right)$
(that is always true due to Lemma \ref{always As is similar to
parallel decomposition}), $\mathop {||}\limits_{i = 1}^n P_i \left(
{A_S } \right)\prec A_S$ (that it is true if and only if $DC1$,
$DC2$ and $DC3$ are true, according to Lemma \ref{Similarity of
parallel decomposition to As}) and $R^{-1}_1 = R_2$ (that for a
deterministic automaton $A_S$, when $A_S\prec \mathop {||}\limits_{i
= 1}^n P_i \left( {A_S } \right)$ with simulation relation $R_1$ and
$\mathop {||}\limits_{i = 1}^n P_i \left( {A_S } \right)\prec A_S$
with simulation relation $R_2$, due to Lemma \ref{symmetric},
$R^{-1}_1 = R_2$ holds true if and only if $DC4$ is satisfied).
Therefore, $A_S\cong \mathop {||}\limits_{i = 1}^n P_i \left( {A_S }
\right)$ if and only if $DC1$, $DC2$, $DC3$ and $DC4$ are satisfied.
\end{proof}

\begin{remark}\label{meaning of DC}
Intuitively, the decomposability condition $DC1$ means that for any
decision on the selection between two transitions there should exist
at least one agent that is capable of the decision making, or the
decision should not be important (both permutations in any order be
legal). $DC2$ says that for any decision on the order of two
successive events before any string, either there should exist at
least one agent capable of such decision making, or the decision
should not be important, i.e., any order would be legal for
occurrence of that string. The condition $DC3$ means that the
interleaving of strings from local task automata that share the
first appearing shared event ($p_{E_i  \cap E_j } \left( s \right)$
and $p_{E_i  \cap E_j } \left( s^{\prime} \right)$  start with the
same
 event $a \in E_i\cap E_j$), should not allow a string that is
not allowed in the original task automaton. In other words, $DC3$ is
to ensure that an illegal behavior (a string that does not appear in
$A_S$) is not allowed by the team (does not appear in $\mathop
{||}\limits_{i = 1}^n P_i \left( {A_S } \right)$). The last
condition, $DC4$, deals with the nondeterminism of local automata.
Here, $A_S$ is deterministic, whereas $P_i \left( {A_S } \right)$
could be nondeterministic. $DC4$ ensures the determinism of
bisimulation quotient of local task automata, in order to guarantee
that the simulation relations from $A_S$ to $\mathop {||}\limits_{i
= 1}^n P_i \left( {A_S } \right)$ and vice versa are inverse of each
other. By providing this property, $DC4$ guarantees that a legal
behavior (appearing in $A_S$) is not disabled by the team (appears
in $\mathop {||}\limits_{i = 1}^n P_i \left( {A_S } \right)$).
\end{remark}

Example \ref{hierarchical algorithm is sufficient} showed a decomposable automaton. Following example illustrate the
automata that are indecomposable due to violation of one of the
decomposability conditions $DC1$-$DC4$, respectively, although
satisfy other three conditions.
\begin{example}\label{undecomposable automata}
The automata
$A_1$: \xymatrix@R=0.1cm{
             \ar[r]&  \bullet \ar[dr]_{e_2} \ar[r]^{e_1} \ar[ld]^{e_3}&\bullet  \\
             \bullet&& \bullet }
             with local event sets $E_1 = \{e_1, e_3\}$,
             $E_2 = \{e_2\}$,  $E_3 = \{e_3\}$;
             $A_2$: \xymatrix@R=0.1cm{
                \ar[r]&  \bullet\ar[dr]_{e_2} \ar[r]^{e_1}&  \bullet \ar@(lu,ru)^{e_2}\ar[r]^{a}&  \bullet \\
             && \bullet \ar[r]^{e_1}&  \bullet \ar[r]^{a}&  \bullet}
             with $E_1 = \{a, e_1\}$, $E_2 = \{a, e_2\}$;
$A_3$: \xymatrix@R=0.5cm{ &   \bullet\ar[r]^{e_2}& \bullet
\ar[r]^{a}& \bullet \ar[r]^{b}& \bullet \\
\ar[r]&  \bullet \ar[u]^{e_1} \ar[d]_{e_2}\ar[r]^{a}& \bullet \ar[r]^{b}& \bullet \ar[r]^{e_2}& \bullet \\
 &   \bullet\ar[r]^{e_1}& \bullet \ar[r]^{a}& \bullet \ar[r]^{b}& \bullet
                }
with $E_1=\{a, b,  e_1\}$, $E_2=\{a, b, e_2\}$, $E_3=\{b\}$, and
    $A_4$: \xymatrix@R=0.1cm{
&  & \bullet\ar[r]^{e_2}& \bullet \ar[r]^{b}& \bullet \\
\ar[r]&  \bullet \ar[ur]^{a} \ar[dr]_{e_1} & &  &               \\
 &  & \bullet\ar[r]^{a}& \bullet \ar[r]^{e_2}& \bullet \ar[r]^{b}& \bullet \ar[r]^{e_3}& \bullet
                }  with $E = E_1\cup E_2\cup E_3$, $E_1 = \{a, b, e_1, e_2, e_3\}$,
                $E_2 = E_3 =\{a, b, e_2, e_3\}$ are not decomposable as
they respectively do not satisfy $DC1$, $DC2$, $DC3$ and $DC4$, while fulfill other three conditions.
\end{example}

\begin{remark}\label{finite string}(Decidability of the conditions)
Since this work deals with finite state automata, the expression
$s\in E^*$ in the decomposability conditions can be checked over
finite states as follows.

The first condition $DC1$ involves no expression ``$s\in E^*$'', and
hence, can be checked over the finite number of states and
transitions. According to the definition, the second condition $DC2$: $\forall e_1, e_2 \in E, q
\in Q, s \in E^* ,\forall E_i \in \left\{ {E_1 ,...,E_n }
\right\},\left\{ {e_1 ,e_2 } \right\} \not\subset E_i $: $\delta
\left( {q,e_1 e_2 s} \right)! \Leftrightarrow \delta \left( {q,e_2
e_1 s}
 \right)!$;
 (or $DC2$: $\forall e_1, e_2 \in E,  q\in Q$, $s\in E^*$: $[\delta(q,
e_1e_2s)!\vee \delta(q, e_2e_1s)!] \Rightarrow [\exists E_i\in\{E_1,
\ldots, E_n\}, \{e_1, e_2\}\subseteq E_i]\vee [ \delta(q,
e_1e_2s)!\wedge \delta(q, e_2e_1s)!]$) can be checked by showing the
existence of a relation $\hat{R}_2$ on the states reachable from
$\delta(q, e_1e_2)$ and $\delta(q, e_2e_1)$ as $(\delta(q, e_1e_2),
\delta(q, e_2e_1))\in \hat{R}_2$, $\forall (q_1, q_2) \in
\hat{R}_2$, $e\in E$:
\begin{enumerate}
\item $\delta(q_1, e) = q_1^{\prime}\Rightarrow \exists q_2^{\prime}\in
Q$, $\delta(q_2, e) = q_2^{\prime}$, $(q_1^{\prime}, q_2^{\prime})
\in \hat{R}_2$, and
\item $\delta(q_2, e) = q_2^{\prime}\Rightarrow \exists q_1^{\prime}\in
Q$, $\delta(q_1, e) = q_1^{\prime}$, $(q_1^{\prime}, q_2^{\prime})
\in \hat{R}_2$.
\end{enumerate}
For instance, $A_2$ in Example \ref{undecomposable automata}
violates $DC2$ as $(\delta(q_0, e_1e_2), \delta(q_0, e_2e_1)) \in
\hat{R}_2$, $\exists e_2\in E$, $\delta(\delta(q_0, e_1e_2), e_2)!$,
but $\neg\delta(\delta(q_0, e_2e_1), e_2)!$.

Checking $DC3$ also can be done over finite states by corresponding
the pairs of strings $s, s^{\prime}$ such that $\exists q, q_1, q_2 \in Q$, $\delta(q, s)= q_1$, $\delta(q,
s^{\prime})= q_2$, $\exists i, j \in \{1, \cdots, n\}$, $i \neq j$, $p_{E_i\cap E_j}(s)$, $p_{E_i\cap E_j}(s^{\prime})$ start with $a\in E_i\cap E_j$, and then forming $A:=\xymatrix@R=0.1cm{
                \ar[r]&  \bullet\ar[r]^{s}\ar[dr]_{s^{\prime}}&  \bullet \\
                &&\bullet}$ and $A_S(q)$ ( an automaton that is obtained from $A_S$, starting from $q$).
and checking
 $\mathop {||}\limits_{i = 1}^n P_i \left( {A} \right)\prec
A_S(q)$. For example, consider $A_3$ in Example \ref{undecomposable automata} and let $s_1$, $s_2$ and $s_3$ denote its strings from top to bottom. This automaton is not decomposable since $\mathop {||}\limits_{i = 1}^n P_i \left( {A} \right)\nprec
A_S(q_0)$ for $A:=\xymatrix@R=0.1cm{
                \ar[r]&  \bullet\ar[r]^{s_1}\ar[dr]_{s_2}&  \bullet \\
                &&\bullet}$.  Here, $s_1$ and $s_2$ synchronize on $a\in E_1\cap E_2$ and generate a new string $e_1abe_2$ in $\mathop {||}\limits_{i = 1}^n P_i \left( {A} \right)$ that does not appear in $A_S$.
The fourth condition ($DC4$: $\forall i\in\{1,...,n\}$, $x, x_1, x_2
\in Q_i$, $x_1\neq x_2$, $e\in E_i$, $t\in E_i^*$, $x_1\in\delta_i
(x, e)$, $x_2\in\delta_i (x, e)$: $\delta_i (x_1, t)!
\Leftrightarrow \delta_i(x_2, t)!$) also can be checked over finite
states, by checking the existence of a relation $\hat{R}_4$ on the
states reachable from $x_1$ and $x_2$ as $(x_1, x_2)\in \hat{R}_4$,
$\forall (x_3, x_4) \in \hat{R}_4$, $e\in E$:
\begin{enumerate}
\item $x_3^{\prime}\in\delta_i(x_3, e)\Rightarrow \exists x_4^{\prime}\in
Q_i$, $x_4^{\prime}\in \delta_i(x_4, e)$, $(x_3^{\prime},
x_4^{\prime}) \in \hat{R}_4$, and
\item $x_4^{\prime}\in\delta_i(x_4, e)\Rightarrow \exists x_3^{\prime}\in
Q_i$, $x_3^{\prime}\in\delta_i(x_3, e)$, $(x_3^{\prime},
x_4^{\prime}) \in \hat{R}_4$.
\end{enumerate} Definition of this relation is a direct implication of $DC4$ that requires identical strings after any nondeterministic transition in any local automaton. For example, the task automaton $A_4$ in Example \ref{undecomposable automata} does not satisfy $DC4$, as for
$P_2(A_S)\cong P_3(A_S)\cong \xymatrix@R=0.1cm{
&  & *+[o][F]{y_1}\ar[r]^{e_2}& *+[o][F]{y_2} \ar[r]^{b}& *+[o][F]{y_3} \\
\ar[r]&  *+[o][F]{y_0} \ar[ur]^{a} \ar[dr]_{a}  \\
 &  &  *+[o][F]{y_4} \ar[r]^{e_2}& *+[o][F]{y_5} \ar[r]^{b}& *+[o][F]{y_6} \ar[r]^{e_3}& *+[o][F]{y_7}
                }$,\\ $\hat{R}_4 = \{(y_1, y_4), (y_2, y_5), (y_3,
                y_6)\}$, $(y_3,
                y_6)\in \hat{R}_4$, $\exists e_3\in E$, $\delta_2(y_6, e_3)!$,
                but $\neg\delta_2(y_3, e_3)!$.

\end{remark}

\begin{remark}\label{Complexity}(Complexity)
Let $|Q^{\prime}|$ be the
summation of number of states in two longest branches of $A_S$ and $|Q|$, $|E|$ and $n$
denote the size of the state space, the size of the event set and
the number of agents (number of local event sets), respectively.

The complexity of $DC1$ is of the order of $|E|^2\times |Q|$, as the pairs of events have to be checked (
$O\left( {\left( {\begin{array}{*{20}c}
   {\left| E \right|}  \\
   2  \\
\end{array}} \right)} \right) = O\left( {\frac{{\left| E \right|!}}{{2\left( {\left| E \right| - 2} \right)!}}} \right) = O\left( {\frac{{\left| E \right|\left( {\left| E \right| - 1} \right)}}{2}} \right) \approx O\left( {\left| E \right|^2 } \right)$) from each state ($|Q|$). Complexity of $DC2$ is calculated as the order of $|E|^2\times |Q| \times |\delta|= |E|^3\times |Q|^3$ as investigating pairs of events from each state is of the order of $|E|^2\times |Q|$ as discussed for $DC1$ and the cardinality of the relation $\delta$ in the worst case is $|\delta|_{max} = |Q|\times |E|\times |Q|$ due to the checking of  events from pairs of states in
$\hat{R}_2$.
The complexity of $DC3$ on the other hand is of the order of $(n\times |E|\times |Q^{\prime}|+|Q^{\prime}|^n\times |E|+|Q^{\prime}|^{2n}\times |E|)\times |E|^2\approx  |Q^{\prime}|^{2n}\times |E|^3$, where $n\times |E|\times |Q^{\prime}|$ is for the natural projections; $|Q^{\prime}|^n\times |E|$ is because of parallel compositions; $|Q^{\prime}|^{2n}\times |E|$ is for checking $\mathop {||}\limits_{i = 1}^n P_i \left( {A } \right)\prec A_S(q)$, and $|E|^2$ is due to picking the pairs of strings as it was discussed for $DC1$. Finally, $DC4$ has the complexity of the order of $n\times |E|\times |Q|+ n\times |Q|^2\times |E|\approx n\times |Q|^2\times |E|$, where the first term is due to checking of each event from each state in each agent, and the second one comes from the checking of each event from pairs of states for each agent in $\hat{R}_2$.

The complexity of the direct method for decomposability, i.e., obtaining the natural projections, doing parallel composition and comparing with the original automaton, has the order of $ n\times|E|\times |Q| + |Q|^n\times |E|+|Q|\times |E|\times |Q|+|Q|^n\times |E|\times |Q|^n+|Q|\times |Q|^n\approx |Q|^{2n}\times |E|$, where the first term is due to the natural projection for each agent, the second one because of parallel composition, the third and fourth terms for checking the simulation relations $A_S\prec \mathop {||}\limits_{i = 1}^n P_i \left( {A_S }
\right)$ and $\mathop {||}\limits_{i = 1}^n P_i \left( {A_S }
\right)\prec A_S$, and the last term is for checking that the simulation relations are inverse of each
other.

Therefore, the complexity of the proposed method is
$|Q^{\prime}|^{2n}\times |E|^3$ while the complexity of the
method with constructing the parallel composition of the natural
projections and checking the bisimilarity with the initial automaton
is of the order $|Q|^{2n}\times |E|$.
In practice, $|Q^{\prime}|\ll |Q|$ and hence for large scale systems with a big $n$, the proposed method yields less complexity.
\end{remark}

More importantly, the proposed method provides some guideline on the
structure of the global specification automaton and the distribution
the events among the agents in order for decomposability.

\begin{remark}\label{Hints}(Insights on enforcing the
decomposability conditions) The result in Theorem
\ref{Decomposability Theorem for n agents} provides us some hints
for ruling out indecomposable task automata and for enforcing the
violated decomposability conditions. For example, if $\exists e_1,
e_2 \in E, q\in Q$: $[\delta(q,e_1)!\wedge \delta(q,e_2)!]$ but
neither $\exists E_i\in\{E_1, \ldots, E_n\}, \{e_1, e_2\}\subseteq
E_i$ nor $\delta(q, e_1e_2)! \wedge \delta(q, e_2e_1)!$, then $A_S$
is not decomposable due to the violation of $DC1$. To remove this
violation there should exist an agent with local event set
$E_i\in\{E_1, \ldots, E_n\}$ such that $\{e_1, e_2\}\subseteq E_i$.
For instance, for $A_1$ in
Examples \ref{undecomposable automata} if
$E_2 = \{e_1, e_2\}$ and $E_3 = \{e_2, e_3\}$, then $DC1$ was satisfied.
This solution also works for an indecomposability of $A_S$ due to a
violation of $DC2$ where $\exists e_1, e_2 \in E,  q\in Q$, $s\in
E^*$: $\delta(q, e_1e_2s)!\vee \delta(q, e_2e_1s)!$ but neither
$\exists E_i\in\{E_1, \ldots, E_n\}, \{e_1, e_2\}\subseteq E_i$ nor
$\delta(q, e_1e_2s)!\wedge \delta(q, e_2e_1s)!$.  Violation of other two conditions,
$DC3$ and $DC4$, is caused due to synchronization of two different
branches $s$ and $s^{\prime}$ from different local task automata,
say $P_i(A_S)$ and $P_j(A_S)$, on a common event $a\in E_i\cap E_j$.
This synchronization may impose an ambiguity in understanding of
$A_S$, when $P_i(A_S)$ and $P_j(A_S)$ synchronize on $a$. If one
string in $P_i(A_S)$ after synchronization on $a$, continues to
another string in $P_j(A_S)$ and this interleaving generates a new
string in $\mathop {||}\limits_{i = 1}^n P_i \left( {A_S } \right)$
that does not appear in $A_S$, then $DC3$ is dissatisfied, whereas
if this interleaving causes that a string in $A_S$ cannot be
completed in $\mathop {||}\limits_{i = 1}^n P_i \left( {A_S }
\right)$, then $DC4$ is violated. $DC4$ can be also violated due to
a nondeterminism on a private event in a local automaton, which
again causes an ambiguity in the collective task $\mathop
{||}\limits_{i = 1}^n P_i \left( {A_S } \right)$. One way to remove
this ambiguity is therefore by introducing the first events in $s$
and $s^{\prime}$ to both $E_i$ and $E_j$. In this case the
synchronization on event $a$ will only occur on the projections of
identical strings from $A_S$ and also it avoids the nondeterminism
in local automata. For example, the task automaton $A_S$:
\xymatrix@R=0.3cm{
                \ar[r]&  \bullet\ar[d]_{a} \ar[r]^{e_1}&  \bullet \ar[r]^{a}&  \bullet  \\
             & \bullet \ar[r]^{e_2}\ar[dr]_{e_3}&  \bullet \ar[r]^{e_3}&  \bullet\\
             &&\bullet \ar[ur]_{e_2}}, with local event sets $E_1 =
             \{a, e_1, e_3\}$ and $E_2 =
             \{a, e_2\}$ satisfies $DC1$ and $DC2$, but violates
             $DC3$ and $DC4$, and hence is not decomposable as the
             parallel composition of $P_1(A_S)\cong \xymatrix@R=0.3cm{
                \ar[r]&  \bullet\ar[dr]_{a} \ar[r]^{e_1}&  \bullet \ar[r]^{a}&  \bullet \\
             && \bullet \ar[r]^{e_3}&  \bullet }$, and $P_2(A_S)\cong \xymatrix@R=0.3cm{
                \ar[r]&  \bullet\ar[d]_{a} \ar[r]^{a}&  \bullet  \\
             &\bullet \ar[r]^{e_2}&  \bullet }$, is
             $P_1(A_S)||P_2(A_S)\cong \xymatrix@R=0.3cm{
             \bullet& \bullet \ar[l]_{a} \ar[d]^{a}& \check{\bullet}
             \ar[l]_{e_1}\ar[r]^{a}\ar[d]^{a}& \bullet \ar[r]^{e_2}\ar[dr]_{e_3}&\bullet
             \ar[r]^{e_3}&\bullet\\
             \bullet& \bullet \ar[l]_{e_2} &\bullet \ar[r]^{e_3} &\bullet &\bullet\ar[ur]_{e_2}
             } \ncong A_S$.
Now, inclusion of $e_1$ in $E_2$ leads to $P_2(A_S)\cong
\xymatrix@R=0.3cm{
                \ar[r]&  \bullet\ar[dr]_{a} \ar[r]^{e_1}&  \bullet \ar[r]^{a}&  \bullet \\
             && \bullet \ar[r]^{e_2}&  \bullet }$ and makes $A_S$
             decomposable.
\end{remark}

Once the task is decomposed into local tasks and the local
controllers are designed for each local plant, the next question is
guaranteeing the global specification, provided each local closed
loop system satisfies its corresponding local specification.

The cooperative tasking result can be now presented as follows.
\begin{theorem}\label{Top-down}
Consider a plant, represented by a parallel distributed system
$\overset{n}{\underset{i=1}{\parallel} }A_{P_i}$, with given local
event sets $E_i$, $i=1,...,n$, and let the global specification is
given by a deterministic task automaton $A_S$, with
$E=\overset{n}{\underset{i=1}{\cup} }E_i$. Then, designing local
controllers $A_{C_i}$, so that $A_{C_i}\parallel A_{P_i}\cong
P_i(A_S)$, $i=1,\cdots, n$, derives the global closed loop system to
satisfy the global specification $A_S$, in the sense of
bisimilarity, i.e., $\overset{n}{\underset{i=1}{\parallel}
}(A_{C_i}\parallel A_{P_i})\cong A_S$, provided $DC1$, $DC2$, $DC3$
and $DC4$ for $A_S$.
\end{theorem}

\begin{proof}
Following two lemmas are used
during the proof.
\begin{lemma}\label{associativity-parallel}(Associativity
of parallel composition \cite{Cassandras2008}) $P_1(A_S)\parallel
P_2(A_S) \parallel \cdots
\parallel
P_{n-1}(A_S)\parallel P_n(A_S)\cong P_n(A_S)\parallel
\left(P_{n-1}(A_S)
\parallel \left( \cdots \parallel \left( P_2(A_S)\parallel
P_1(A_S)\right)\right)\right)$.
\end{lemma}

\begin{lemma}\label{ParallelBiSimulation}\cite{Automatica2010-2-agents-decomposability}
If two automata $A_2$ and $A_4$ (bi)simulate, respectively, $A_1$
and $A_3$, then $A_2\parallel \  A_4$ (bi)simulates $A_1\parallel \
A_3$, i.e.,
\begin{enumerate}
\item $\left(A_1\prec A_2 \right)\wedge \left(A_3\prec
A_4\right)\Rightarrow \left(A_1\parallel \ A_3\prec A_2\parallel \
A_4\right)$;
\item $\left(A_1\cong A_2 \right)\wedge \left(A_3\cong
A_4\right)\Rightarrow \left(A_1\parallel \ A_3\cong A_2\parallel \
A_4\right)$.
\end{enumerate}
 \end{lemma}

Now, satisfying $DC1$-$DC4$ for $A_S$, according to Theorem
\ref{Decomposability Theorem for n agents}, leads to decomposability
of $A_S$ into local task automata $P_i(A_S)$, $i=1,...,n$, such that
$A_S \cong \overset{n}{\underset{i=1}{\parallel} }P_i(A_S)$. Then,
choosing local controllers $A_{C_i}$, so that $A_{C_i}\parallel
A_{P_i}\cong P_i(A_S)$, $i=1,2,\cdots,n$, due to Lemmas $\ref{associativity-parallel}$ and
$\ref{ParallelBiSimulation}.2$, results in
$\overset{n}{\underset{i=1}{\parallel} }(A_{C_i}\parallel
A_{P_i})\cong \overset{n}{\underset{i=1}{\parallel} }P_i(A_S)\cong
A_S$.
\end{proof}

Now, if $DC1$-$DC4$ is reduced to $DC1$-$DC3$ (conditions in Theorem
\ref{Decomposability Theorem for n agents} are reduced into the conditions
in Lemma \ref{Similarity of parallel decomposition to As}), then
$\overset{n}{\underset{i=1}{\parallel} }P_i(A_S)\cong A_S$ is
reduced into $\overset{n}{\underset{i=1}{\parallel} }P_i(A_S)\prec
A_S$, and hence, choosing local controllers $A_{C_i}$, so that
$A_{C_i}\parallel A_{P_i}\prec P_i(A_S)$, $i=1,2,\cdots,n$, due to
Lemmas $\ref{associativity-parallel}$ and  $\ref{ParallelBiSimulation}.1$ leads to
$\overset{n}{\underset{i=1}{\parallel} }(A_{C_i}\parallel
A_{P_i})\prec \overset{n}{\underset{i=1}{\parallel} }P_i(A_S)\prec
A_S$. Therefore,
\begin{corollary}
Considering the plant and global task as stated in Theorem
\ref{Top-down}, if $DC1$-$DC3$ are satisfied, then designing local
controllers $A_{C_i}$, so that $A_{C_i}\parallel A_{P_i}\prec
P_i(A_S)$, $i=1,\cdots, n$, derives the global closed loop system to
satisfy the global specification $A_S$, in the sense of similarity,
i.e., $\overset{n}{\underset{i=1}{\parallel} }(A_{C_i}\parallel
A_{P_i})\prec A_S$.
\end{corollary}

In the following example, we recall the task automaton of cooperative multi-robot scenario from \cite{Automatica2010-2-agents-decomposability} (with the correction of robot indices $R_2$, $R_1$ and $R_3$ from right to the left), where the global task automaton has been decomposed into local task automata using a hierarchical approach as a sufficient condition by which the decomposability conditions for $2$ agents are successively used for $n$ agents.
Here, we decompose $A_S$ directly using Theorem
\ref{Decomposability Theorem for n agents}.
\begin{example}(Revisiting Example in Section $5$
for decomposability using Theorem \ref{Decomposability Theorem for n agents}) \label{MRS example-n-agents}
Consider a team of three robots $R_1$, $R_2$ and $R_3$ in
Figure $\ref{MRS-Cooperation}$, initially in Room $1$.
\begin{figure}[ihtp]
      \begin{center}
     \includegraphics[width=0.5\textwidth]{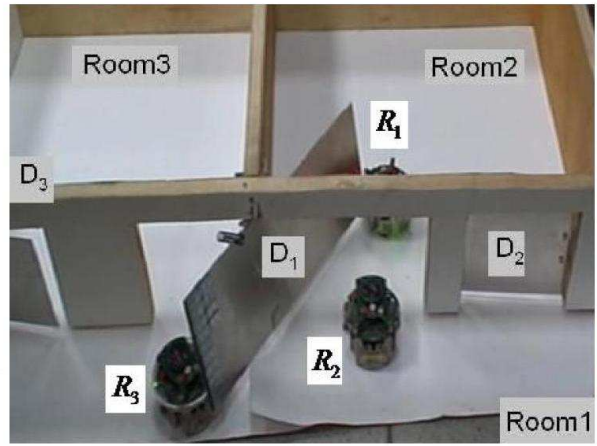}
        \caption{The environment of MRS coordination example.}
         \label{MRS-Cooperation}
        \end{center}
\end{figure}
 All doors are equipped with spring to be
closed automatically, when there is no force to keep them open.
After a help announcement from Room $2$, the Robot $R_2$ is required to go to Room $2$,
urgently from the one-way door $D_2$ and accomplish its task there and come back
immediately to Room $1$ from the two-way, but heavy door $D1$ that needs the cooperation of two robots $R_1$ and
$R_3$ to be opened. To save time, as soon as the robots hear the
help request from Room $2$, $R_2$ and $R_3$ go to Rooms $2$ and $3$,
from $D_2$ and the two-way door $D_3$, respectively, and then $R_1$ and $R_3$
position on $D_1$, synchronously open $D_1$ and wait for the
accomplishment of the task of $R_2$ in Room $2$ and returning to
Room $1$ ($R_2$ is fast enough). Afterwards, $R_1$ and $R_3$ move
backward to close $D_1$ and then $R_3$ returns back to Room $1$ from
$D_3$. All robots then stay at Room $1$ for the next
 task \cite{Automatica2010-2-agents-decomposability}.
These requirements can be translated into a task automaton for the
robot team as it is illustrated in Figure $\ref{global task
automaton}$, defined over local event sets $E_1 = \{h_1, R_1toD_1,
R_1onD_1, FWD, D_1opened, R_2in1, BWD, D_1closed, r\}$, $E_2= \{h_2,
R_2to2, R_2in2, D_1opened, R_2to1, R_2in1, r\}$, and $E_3 = \{h_3$,
$R_3to3$, $R_3in3$, $R_3toD_1$, $R_3onD_1$, $FWD$, $D_1opened$,
$R_2in1$, $BWD$, $D_1closed$, $R_3to1$, $R_3in1$, $r\}$, with
$h_i$:= $R_i$ received help request, $i= 1, 2, 3$; $R_jtoD_1$:=
command for $R_j$ to position on $D_1$, $j= 1, 3$; $R_jonD_1$:=
$R_j$ has positioned on $D_1$, $j= 1,3$; $FWD$:= command for moving
forward (to open $D_1$); $BWD$:= command for moving backward (to
close $D_1$); $D_1opened$:= $D_1$ has been opened;  $D_1closed$:=
$D_1$ has been closed; $r$:= command to go to initial state for the
next implementation; $R_itok$:= command for $R_i$ to go to Room $k$,
and $R_iink$:= $R_i$ has gone to Room $k$, $i = 1, 2, 3$, $k = 1, 2,
3$.

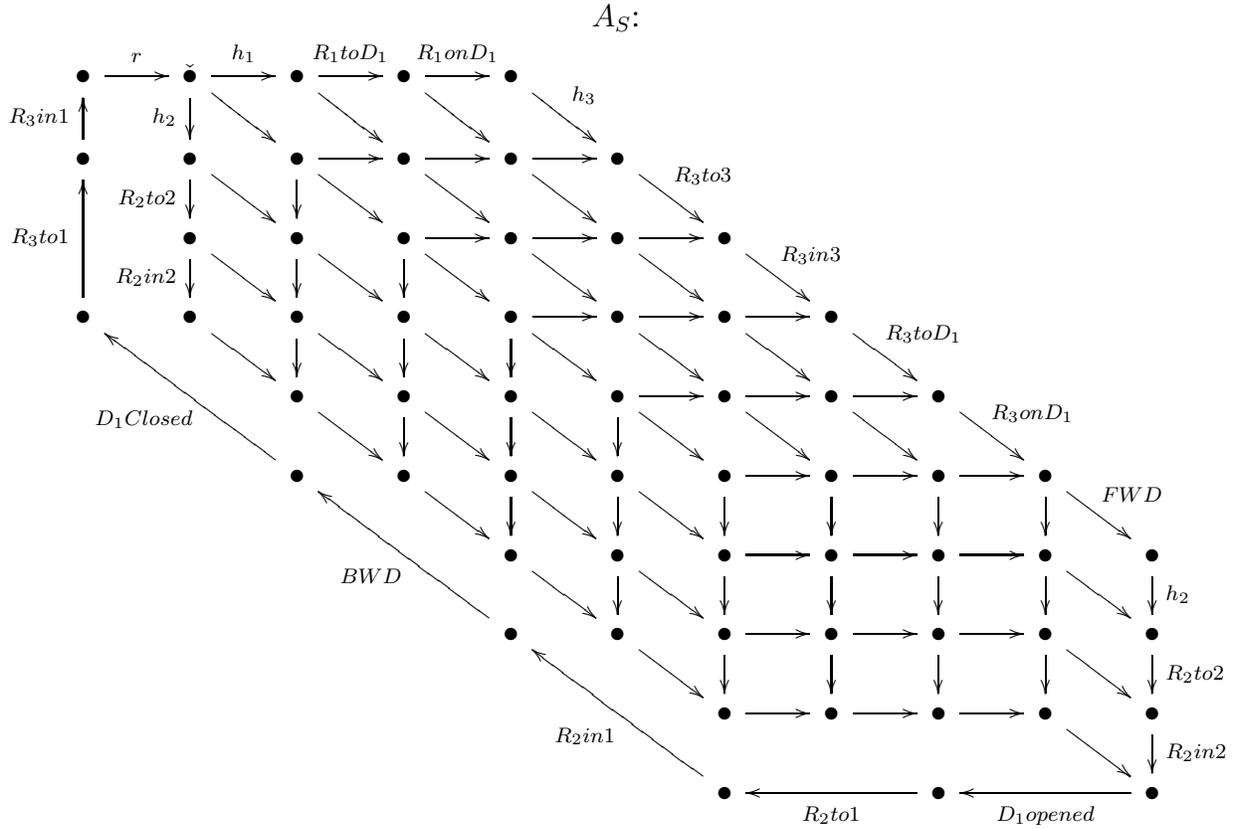
\begin{figure}[ihtp]
      \begin{center}
     $A_S$:\\ \xymatrix@R=0.5cm{
 \bullet  \ar[r]^{r}  & \check{\bullet}\ar[r]^{h_1}\ar[d]_{h_2}\ar[dr]
 &\bullet \ar[r]^{R_1toD_1}\ar[dr] &\bullet \ar[r]^{R_1onD_1}\ar[dr] &\bullet \ar[dr]^{h_3}\\
\bullet \ar[u]^{R_3in1}& \bullet \ar[d]_{R_2to2}\ar[dr]&\bullet
\ar[d]\ar[dr]\ar[r]&\bullet \ar[r]\ar[dr] &\bullet
\ar[r]\ar[dr] &\bullet \ar[dr]^{R_3to3}\\
& \bullet \ar[d]_{R_2in2}\ar[dr] &\bullet \ar[d]\ar[dr]& \bullet
\ar[r]\ar[d]\ar[dr]
 &\bullet \ar[r]\ar[dr]
  &\bullet \ar[r]\ar[dr]
  &\bullet \ar[dr]^{R_3in3}\\
 \bullet
\ar[uu]^{R_3to1}&\bullet\ar[dr]&\bullet\ar[d]\ar[dr]
&\bullet\ar[d]\ar[dr] &\bullet\ar[d]\ar[dr]\ar[r]
&\bullet\ar[r]\ar[dr] &\bullet\ar[r]\ar[dr]
&\bullet\ar[dr]^{R_3toD_1}\\
&&\bullet\ar[dr]&\bullet\ar[d]\ar[dr] &\bullet\ar[d]\ar[dr]
&\bullet\ar[d]\ar[r]\ar[dr] &\bullet\ar[r]\ar[dr]
&\bullet\ar[r]\ar[dr]
&\bullet\ar[dr]^{R_3onD_1}\\
&&\bullet\ar[uull]^{D_1Closed}&\bullet\ar[dr] &\bullet\ar[d]\ar[dr]
&\bullet\ar[d]\ar[dr] &\bullet\ar[d]\ar[r] &\bullet\ar[d]\ar[r]
&\bullet\ar[d]\ar[r]
&\bullet\ar[d]\ar[dr]^{FWD}\\
&&&&\bullet\ar[dr] &\bullet\ar[d]\ar[dr] &\bullet\ar[d]\ar[r]
&\bullet\ar[d]\ar[r] &\bullet\ar[d]\ar[r]
&\bullet\ar[d]\ar[dr] &\bullet\ar[d]^{h_2}\\
&&&&\bullet\ar[uull]^{BWD} &\bullet\ar[dr] &\bullet\ar[d]\ar[r]
&\bullet\ar[d]\ar[r] &\bullet\ar[d]\ar[r]
&\bullet\ar[d]\ar[dr] &\bullet\ar[d]^{R_2to2}\\
 &&&&&&\bullet\ar[r]
&\bullet\ar[r]&\bullet\ar[r]&
\bullet\ar[dr]& \bullet\ar[d]^{R_2in2}\\
&&&&&&\bullet\ar[uull]^{R_2in1}
&&\bullet\ar[ll]^{R_2to1}&&\bullet\ar[ll]^{D_1opened}
                }
        \caption{Task automaton $A_S$ for robot team.}
\label{global task automaton}
        \end{center}
      \end{figure}

To check the decomposability of $A_S$ using Theorem \ref{Decomposability Theorem for n agents},
firstly $DC1$ and $DC2$
are satisfied since for any order/selection on the pairs events, each
from one of the sets $\{h_1, R_1toD_1, R_1onD_1\}\subseteq
E_1\backslash \{E_2\cup E_3\}$, $\{h_2, R_2to2, R_2in2\}\subseteq
E_2\backslash \{E_1\cup E_3\}$ and $\{h_3, R_3to3, R_3in3, R_3toD_1,
R_3onD_1\}\subseteq E_3\backslash \{E_1\cup E_2\}$ and also the
pairs of event $FW$, paired with events from $\{h_2, R_2to2,
R_2in2\}$, the events appear in both orders in the automaton. The rest of
orders/selections on transitions that are not legal in both orders
can be decided by at least one agent, as $\{R_1onD_1, FWD\}\subseteq
E_1$, $\{R_3onD_1, FWD\}\subseteq E_3$, $\{FWD, D_1opened\}\subseteq
E_1$, $\{D_1opened, R_2to1\}\subseteq E_2$, $\{R_2to1,
R_2in1\}\subseteq E_2$, $\{R_2in1, BWD\}\subseteq E_1$, $\{BWD,
D_1closed\}\subseteq E_1$, $\{D_1closed, R_3to1\}\subseteq E_3$,
$\{R_3to1, R_3in1\}\subseteq E_3$, $\{R_3in1, r\}\subseteq E_3$,
$\{r, h_1\}\subseteq E_1$, $\{r, h_2\}\subseteq E_2$, $\{r,
h_3\}\subseteq E_3$. Moreover, since starting from any state, each
shared event $e\in\{FWD, D_1opened, R_2in1, BWD, D_1closed, r\}$
appears in only one branch, $DC3$ is satisfied. Furthermore, $DC4$
is also satisfied since $P_i(A_S)$, $i = 1, 2, 3$ are deterministic
automata. Therefore, according to Theorem \ref{Decomposability
Theorem for n agents}, $A_S$ is decomposable into $P_i(A_S)$, $i =
1, 2, 3$, as illustrated in Figure \ref{The Second
Stage of Decomposition}, bisimulates $A_S$.
\begin{figure}[ihtp]
     $P_1(A_S)$:
      \xymatrix@C=0.5cm{
     \ar[r]&  \bullet \ar[r]^{h_1}&  \bullet\ar[r]_{R_1toD_1}& \bullet\ar[r]^{R_1onD_1}&
     \bullet\ar[r]_{FWD}& \bullet\ar[r]^{D_1opened}& \bullet\ar[r]_{R_2in1}& \bullet\ar[r]^{BWD}&
     \bullet\ar[r]_{D_1closed}& \bullet
\ar`dr_l[llllllll]`_u[llllllll]_{r}[llllllll]
   }\\
$P_2(A_S)$: \xymatrix@C=0.5cm{
     \ar[r]&  \bullet \ar[r]^{h_2} &  \bullet  \ar[r]_{R_2to2}& \bullet \ar[r]^{R_2in2}& \bullet\ar[r]_{D_1opened}&
     \bullet\ar[r]^{R_2to1}& \bullet\ar[r]_{R_2in1}& \bullet
\ar`dr_l[llllll]`_u[llllll]_{r}[llllll]
   }\\
       $P_3(A_S)$:  \xymatrix@C=0.5cm{
     \ar[r]&  \bullet \ar[r]^{h_3}& \bullet \ar[r]_{R_3to3}& \bullet \ar[r]^{R_3in3}&
     \bullet\ar[r]_{R_3toD_1}& \bullet\ar[r]^{R_3onD_1}&
     \bullet\ar[r]_{FWD}& \bullet\ar[r]^{D_1opened}& \bullet\ar[r]_{R_2in1}& \bullet\ar[r]^{BWD}&
     \bullet\ar[r]_{D_1closed}&
     \bullet\ar[r]^{R_3to1}&
     \bullet\ar[r]_{R_3in1}&\bullet
\ar`dr_l[llllllllllll]`_u[llllllllllll]_{r}[llllllllllll]
   }
        \caption{$P_1(A_S)$ for $R_1$; $P_2(A_S)$ for $R_2$ and  $P_3(A_S)$ for $R_3$.}
 \label{The Second Stage of Decomposition}
      \end{figure}
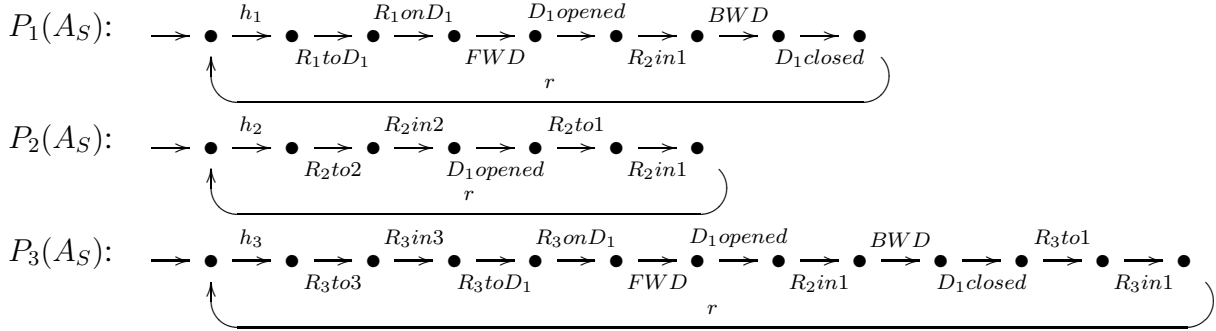
Choosing local controllers $A_{C_i}:= P_i(A_S)$ leads to $A_{C_i}\parallel
A_{P_i}\cong P_i(A_S)$, $i=1, 2, 3$ that according to Theorem \ref{Top-down}
results in $\overset{n}{\underset{i=1}{\parallel} }(A_{C_i}\parallel
A_{P_i})\cong \overset{n}{\underset{i=1}{\parallel} }P_i(A_S)\cong
A_S$, i.e., the team of controlled robots collectively satisfy the global specification $A_S$.
Suppose that $R_1$ does not inform the occupance $D_1opened$ to $R_2$. In that case, there was not exist an agent to decide on the order of event pairs $\{D_1opened, R_2to1\}$ and the task was undecomposable. According to the insight from $DC2$, sharing $D_1opened$ between $R_1$ and $R_2$ makes $A_S$ decomposable.
The scenario has been successfully implemented on a team of three
ground robots. We include a breif version of the example in the paper and, due to the restriction in space,
 the reader are referred to \cite{Automatica2010-2-agents-decomposability} for the description and figures of the scenario and the global task.
\end{example}

\section{CONCLUSION}\label{CONCLUSION}
The paper proposed a formal method for automaton decomposability,
applicable in top-down decentralized cooperative control of
distributed discrete event systems. Given a set of agents whose
logical behaviors are modeled in a parallel distributed system, and
a global task automaton, the paper has the following contributions:
firstly, we provide necessary and sufficient conditions for
decomposability of an automaton with respect to parallel composition
and natural projections into an arbitrary finite number of local
event sets, and secondly, it has been shown that if a global task
automaton is decomposed for individual agents, designing a local
supervisor for each agent, satisfying its local task, guarantees
that the closed loop system of the team of agents satisfies the
global specification.

The proposed decomposability conditions can be applied to the discrete event systems in which all states are marked. The example of such systems include the manufacturing machines with routine tasks, execution of PLC (programmable Logic Controller) systems that the subroutines are visited iteratively, and any other such systems that all states of the system should be visited and hence can be attributed to marked states.
Therefore, future works include the extension of the results for the class of systems with only some of the states as marked states. For this purpose new decomposability conditions have to be developed such that the composition of local automata preserves the marked
states of the global task automaton. Other interesting directions on this topic are the fault-tolerant task decomposition in spite of failure in some events, and decomposabilizability of an
indecomposable task automaton by modifying the event distribution.

\appendices
\section{Definitions}
This part provides some definitions to be used during the proofs of the lemmas in the Appendix.
Firstly, we successive event pair and adjacent
event pair are defined as follows.
\begin{definition}(Successive event pair)
Two events $e_1$ and $e_2$ are called successive events if $\exists
q\in Q: \delta(q,e_1)! \wedge \delta(\delta(q,e_1),e_2)!$ or
$\delta(q,e_2)! \wedge \delta(\delta(q,e_2),e_1)!$.
\end{definition}
\begin{definition}(Adjacent event pair)
Two events $e_1$ and $e_2$ are called adjacent events if $\exists
q\in Q:\delta(q,e_1)! \wedge \delta(q,e_2)!$.
\end{definition}

We will also use synchronized product of languages in the following section, defined as follows.

\begin{definition}(Synchronized product of languages
\cite{Willner1991})\label{Synchronized product of languages}
Consider a global event set $E$ and local event sets $E_i$, $ i = 1,
\ldots, n$, such that $E = \overset{n}{\underset{i=1}{\cup} } E_i$.
For a finite set of languages $\{L_i \subseteq E_i^*\}_{i = 1}^n$,
the synchronized product (product language) of $\{L_i\}$, denoted by
$\overset{n}{\underset{i=1}{|} } L_i$, is defined as
$\overset{n}{\underset{i=1}{|} } L_i = \{s \in E^*|\forall
i\in\{1,\ldots, n\}: p_i(s)\in L_i\} =
\overset{n}{\underset{i=1}{\cap} } p_i^{-1}(L_i)$.
\end{definition}

%
\begin{remark}\label{language of parallel automata-remark}
Using the product language, it is then possible to characterize the
language of parallel composition of two automata $A_1$ and $A_2$,
with respective event sets $E_1$ and $E_2$, in terms of their
languages, as $L(A_1||A_2) = L(A_1)|L(A_2) = p_1^{-1}(L(A_1))\cap
p_2^{-1}(L(A_2))$ with $p_i: (E_1\cup E_2)^* \to E_i^*$, $i = 1, 2$
\cite{Willner1991}. Accordingly, the interleaving of two strings is defined as the product language to their respective automata as
follows. Let $A_1 = (\{q_1,...,q_n\}, \{q_1\}, E_1 =
\{e_1,...,e_n\}, \delta_1)$ and $A_2 =
(\{q_1^{\prime},...,q_m^{\prime}\}, \{q_1^{\prime}\}, E_2 =
\{e_1^{\prime},...,e_m^{\prime}\}, \delta_2)$ denote path automata
(automata with only one branch)
$q_1\overset{e_1}{\underset{}\rightarrow }
q_2\overset{e_2}{\underset{}\rightarrow
}...\overset{e_n}{\underset{}\rightarrow }q_n$ and
$q_1^{\prime}\overset{e_1^{\prime}}{\underset{}\rightarrow }
q_2^{\prime}\overset{e_2^{\prime}}{\underset{}\rightarrow
}...\overset{e_m^{\prime}}{\underset{}\rightarrow }q_m^{\prime}$,
respectively. Then, $L(A_1||A_2) = \bar{s}|\bar{s^{\prime}} =
p_1^{-1}(\bar{s})\cap p_2^{-1}(\bar{s^{\prime}})$ with $s =
e_1,...,e_n$, $s^{\prime} = e_1^{\prime},...,e_m^{\prime}$ and $p_i:
(E_1\cup E_2)^* \to E_i^*$, $i = 1, 2$. Here, $\overline{s}$ denotes the prefix-closure of an string, defined as the set of all prefixes of the string. Formally, if $s$ is the event sequence $s:=e_1e_2...e_n$, then $\overline{s}:= \{\varepsilon, e_1, e_1e_2,..., e_1e_2...e_n\}$.

\begin{example}\label{interleaving-Example}
Consider three strings $s_1 = e_1a$, $s_2 = ae_2$ and $s_3 = ae_1$.
Then the interleaving of $s_1$ and $s_2$ is
$\overline{s_1}|\overline{s_2} = \overline{e_1ae_2}$ while the
interleaving of two strings $s_2$ and $s_3$ becomes
$\overline{s_2}|\overline{s_3} = \overline{\{ae_1e_2, ae_2e_1\}}$.
\end{example}
\end{remark}

\section{Proof for Lemma \ref{always As is similar to parallel
decomposition}}\label{Appendix always As is similar to parallel
decomposition} Recalling Lemma $1$ in
\cite{Automatica2010-2-agents-decomposability}, stating that for a
deterministic automaton $A_S = (Q, q_0, E = E_1\cup E_2, \delta)$,
$A_S \prec P_1(A_S)||P_2(A_S)$, it leads to $P_{\mathop
{\cup}\limits_{i = m}^n E_i}(A_S)\prec P_m(A_S)||P_{\mathop
{\cup}\limits_{i = m+1}^n E_i}(A_S)$, $m = 1,\ldots, n-1$, for $A_S
= (Q, q_0, E = \mathop {\cup}\limits_{i = 1}^n E_i, \delta)$.
Therefore, $A_S \cong  P_{\mathop {\cup}\limits_{i=1}^n
E_i}(A_S)\prec P_1(A_S)||P_{\mathop {\cup}\limits_{i = 2}^n
E_i}(A_S)\prec P_1(A_S)||P_2(A_S)||P_{\mathop {\cup}\limits_{i =
3}^n E_i}(A_S)\prec \ldots \prec \mathop {||}\limits_{i = 1}^n
P_i(A_S)$.

\section{Proof for Lemma \ref{Similarity of parallel decomposition to
As}}\label{Appendix Similarity of parallel decomposition to As}

\textbf{Sufficiency:}
Consider the deterministic automaton $A_S =
(Q, q_0, E, \delta)$. The set of transitions in $\mathop
{||}\limits_{i = 1}^n P_i(A_S) = (Z, z_0, E, \delta_{||})$ is
defined as $T = \{z_0:= (x_0^1,\cdots,x_0^n)\overset{\mathop
{|}\limits_{i = 1}^n \overline{p_i(s_i)}} \longrightarrow z:=
(x_1,\cdots,x_n)\in Z:=\mathop {\prod}\limits_{i = 1}^n Q_i\}$,
where, $(x_0^1,\cdots,x_0^n)\overset{\mathop {|}\limits_{i = 1}^n
\overline{p_i(s_i)}}\longrightarrow(x_1,\cdots,x_n)$ in $\mathop
{||}\limits_{i = 1}^n P_i(A_S)$ is the interleaving of strings
$x_0^i\overset{ p_i(s_i)}\longrightarrow x_i$ in $P_i(A_S)$, $i
=1,\cdots, n$ (projections of $q_0\overset{ s_i }\longrightarrow
\delta(q_0, s_i)$ in $A_S$. Let $\tilde L\left( {A_S } \right) \subseteq L\left( {A_S }
\right)$ denote the largest subset of $L\left( {A_S } \right)$ such that
$\forall s\in \tilde L\left( {A_S } \right), \exists s^{\prime} \in
 \tilde L\left( {A_S } \right)$, $\exists E_i, E_j  \in \left\{ E_1
,...,E_n  \right\},
 i \ne j, p_{E_i  \cap E_j } \left( s \right)$ and $p_{E_i  \cap E_j } \left( s^{\prime} \right)$  start with the same
 event. Then, $T$ can be divided into three sets of
transitions corresponding to a division of $\{\Gamma_1, \Gamma_2,
\Gamma_3\}$ on the set of interleaving strings $\Gamma = \{\mathop
{|}\limits_{i = 1}^n \overline{p_i(s_i)}|s_i \in E^*, q_0\overset{
s_i }\longrightarrow \delta(q_0, s_i)\}$, where, $\Gamma_1 =
\{\mathop {|}\limits_{i = 1}^n \overline{p_i(s_i)}\in \Gamma | s_1,
\cdots, s_n \notin \tilde{L}(A_S), s_1=\cdots = s_n\}$, $\Gamma_2 =
\{\mathop {|}\limits_{i = 1}^n \overline{p_i(s_i)}\in \Gamma | s_1, \cdots, s_n \notin \tilde{L}(A_S),
\exists s_i, s_j \in \{s_1, \cdots, s_n\}, s_i \neq s_j,
\}$, $\Gamma_3 =
\{\mathop {|}\limits_{i = 1}^n \overline{p_i(s_i)}\in \Gamma | s_i
\in \tilde{L}(A_S)\}$. Moreover, since $A_S$ is deterministic, $\mathop {||}\limits_{i = 1}^n P_i(A_S)\prec A_S$ is reduced to $\delta(q_0, \mathop {|}\limits_{i = 1}^n \overline{p_i(s)})!$  in $A_S$ for transitions in $\Gamma$.
$\mathop {||}\limits_{i = 1}^n P_i(A_S)\prec A_S$.

Thus, defining a relation $R$ as $(z_0, q_0)\in
R$, $R:=\{(z, q)\in Z\times Q|\exists t\in E^*, z\in
\delta_{||}(z_0, t)\}$, the aim is to show that $R$ is a simulation
relation from $\mathop {||}\limits_{i = 1}^n P_i(A_S)$ to $A_S$.

For the interleavings in $\Gamma_1$, $\forall z, z_1 \in Z$, $e \in
E$, $z_1\in\delta_{||}(z, e)$: $\exists q, q_1 \in Q$, $\delta (q,
e) = q_1$ such that $\forall z[j]\in\{z[1], \cdots, z[n]\}$ (the
$j-th$ component of $z$), $\exists l \in loc(e)$, $z[j] = [q]_l$.
Similarly, $\forall e ^{\prime}\in E$, $z_2 \in Z$,
$z_2\in\delta_{||}(z_1, e^{\prime})$: $\exists q_2 \in Q$, $\delta (q_1,
e^{\prime}) = q_2$. Now, if $\nexists E_i \in \{E_1, \cdots, E_n\}$,
$\{e, e ^{\prime}\} \in E_i$, then the definition of parallel
composition will furthermore induce that $\exists z_3 \in Z$,
$z_3\in\delta_{||}(z, e ^{\prime})$, $z_2\in\delta_{||}(z_3, e)$.
This, together with $DC1$ and $DC2$ implies that $\exists q_3, q_4
\in Q$, $\delta(q, e ^{\prime}) = q_3$, $\delta(q_3, e) = q_4$ and
that $\forall t\in E^*$, $\delta_{||}(z_2, t)!$: $\delta(q_2, t)!$
and $\delta(q_4, t)!$. Therefore, any path automaton in $\mathop
{||}\limits_{i = 1}^n P_i(A_S)$ is simulated by $A_S$, and hence,
$\delta(q_0, \mathop {|}\limits_{i = 1}^n \overline{p_i(s)})!$ in
$A_S$, $\forall s\in L(A_S)$.

For the interleavings in $\Gamma_2$, from the definition of
$\Gamma_2$, it follows that for any set of $s_i$, $\delta(q_0,
s_i)!$, $i \in \{1,\cdots, n\}$, two cases are possible for
$\Gamma_2$:

\textbf{Case 1:} $\forall s, s^{\prime}\in \{s_1,\cdots, s_n\}$,
$\forall E_i, E_j \in \{E_1,\cdots, E_n\}$: $p_{E_i\cap E_j}(s) =
\varepsilon$ and $p_{E_i\cap E_j}(s^{\prime}) = \varepsilon$. In this case,
projections of such strings $s_i$ can be written as $p_i(s_i) =
e_1^i,\cdots,e_{m_i}^i$, $ i = 1,\cdots, n$. The interleaving of
these projected strings leads to a grid of states and transitions
in $\mathop {\prod}\limits_{i = 1}^n \mathop {\prod}\limits_{j_i =
0}^{m_i} x_{j_i}^i$ as
$(x_{j_1}^{i_1},\cdots,x_{j_n}^{i_n})\overset{ e_j^i
}\longrightarrow (y_{j_1}^{i_1},\cdots,y_{j_n}^{i_n})$, with
$y_{j_i}^{i_k} = \left\{
\begin{array}{ll}
   x_{j_{i+1}}^{i_k}, & \hbox{ if } i = i_k, j = j_i+1\\
   x_{j_i}^{i_k}, & \hbox{otherwise}
\end{array}\right.$
 $j_i =
0,1,\cdots,m_i$, $i = 1,\cdots,n$, $i_k = 1,\cdots,n$, $k =
1,\cdots,n$. This grid of transitions is simulated by counterpart
transitions in $A_S$, as $\forall s, s^{\prime}\in
\{s_1,\cdots,s_n\}$, for any two successive/adjacent events $e_j^i$
and $e_{j^{\prime}}^{i^{\prime}}$, both orders exist in $A_S$, due
to $DC1$ and $DC2$, and hence, $\delta(q_{j_i, i_k},e_j^k) = q_{j_i,
i_k}^{\prime}$, $j_i = 0, 1, \cdots, m_i$, $i = 1, \cdots, n$, $i_k =
1,\cdots,n$, $k = 1,\cdots,n$. Therefore, for any choice of $s_i$
corresponding to $\Gamma_2$, $\delta(q_0, \mathop {|}\limits_{i =
1}^n \overline{p_i(s_i)})!$ in $A_S$.

\textbf{Case 2:} $\exists s, s^{\prime}\in \{s_1,\cdots, s_n\}$,
$\exists E_i, E_j \in \{E_1,\cdots, E_n\}$: $p_{E_i\cap E_j}(s) \neq
\varepsilon$ or $p_{E_i\cap E_j}(s^{\prime}) \neq \varepsilon$, but
they do not start with the same event. Any such $s$ and $s^{\prime}$
can be written as $s = t_1at_2$ and $s^{\prime} =
t_1^{\prime}bt_2^{\prime}$, where $t_1 = e_1\cdots e_m, t_1^{\prime}
= e_1^{\prime}\cdots e^{\prime}_{m^{\prime}} \notin (E_i\cap E_j)^*, \forall i ,j
\in \{1,\cdots,n\}, i \neq j$, $\exists i ,j \in \{1,\cdots,n\}, i
\neq j$, $a, b \in E_i\cap E_j$, $t_2, t_2^{\prime}\in E^*$.
Therefore, due to synchronization constraint, the interleaving of
strings will not evolve from $a$ and $b$ onwards, and hence,
$\overline{p_i(s)}|\overline{p_j(s^{\prime})} =
\overline{p_i(t_1)}|\overline{p_j(t_1^{\prime})}$ and
$\overline{p_i(s^{\prime})}|\overline{p_j(s)} =
\overline{p_i(t_1^{\prime})}|\overline{p_j(t_1)}$, and Case $2$ is
reduced to Case $1$, leading to $\delta(q_0, \mathop {|}\limits_{i =
1}^n \overline{p_i(s_i)})!$ in $A_S$.

Furthermore, due to $DC3$, for any two distinct strings $s, s^{\prime}\in \tilde{L}(A_S)$ (i.e., two strings starting from state $q$ in $A_S$ that $\exists E_i, E_j \in \{E_1,..., E_n\}$, $i\ne j$, $p_{E_i\cap E_j}(s), p_{E_i\cap E_j}(s^{\prime})$ start with the same event $a\in E_i\cap E_j$) we have $\mathop {||}\limits_{i = 1}^n P_i \left( {A} \right)\prec
A_S(q)$ (where $A:=\xymatrix@R=0.1cm{
                \ar[r]&  \bullet\ar[r]^{s}\ar[dr]_{s^{\prime}}&  \bullet \\
                &&\bullet    }$ and $A_S(q)$ denotes an automaton that is obtained from $A_S$, starting from $q$). This is particularly true for $q = q_0$. Therefore, $DC3$ implies that for the pair of strings $s, s^{\prime}$ (over the transitions in $\Gamma_3$), and corresponding automaton $A$, $L(\mathop {||}\limits_{i = 1}^n P_i \left( {A} \right))\subseteq L(A_S)$, that from the definition of synchronized product means that $\overset{n}{\underset{i=1}{\cap} } p_i^{-1}(\{\bar{s}, \bar{s^{\prime}}\})\subseteq L(A_S)$. For any pair of $s^{\prime}, s^{\prime\prime}\in \bar{L}(A_S)$ also $DC3$ similarly results in $\overset{n}{\underset{i=1}{\cap} } p_i^{-1}(\{\bar{s^{\prime}}, \bar{s^{\prime\prime}}\})\subseteq L(A_S)$, that collectively results in $\overset{n}{\underset{i=1}{\cap} } p_i^{-1}(\{\bar{s}, \bar{s^{\prime}}, \bar{s^{\prime\prime}}\})\subseteq L(A_S)$, due to the following lemma:
                \begin{lemma}\cite{Cassandras2008}
                For any two languages $L_1, L_2$ defined over an event set $E$ and a natural projection $p:E*\to E_i^*$, for $E_i\subseteq E$:
                $p_i(L_1\cup L_2) = p_i(L_1)\cup p_i(L_2)$ and $p_i^{-1}(L_1\cup L_2) = p_i^{-1}(L_1)\cup p_i^{-1}(L_2)$.
                \end{lemma}
                This, inductively means that for $\{s_1 \cdots, s_m\}\subseteq \tilde{L}(A_S)$: $\overset{n}{\underset{i=1}{\cap} } p_i^{-1}(\{s_i\}_{i = 1}^m)\subseteq L(A_S)$, i.e., $\delta(q_0, \mathop {|}\limits_{i =
1}^n \overline{p_i(s_i)})!$ in $A_S$, for transitions in $\Gamma_3$.

Therefore, $DC3$ implies that all transitions in $\Gamma$ are simulated by transitions in $A_S$ that because of the determinism of $A_S$ results in
$\mathop {||}\limits_{i = 1}^n P_i \left( {A_S} \right)\prec
A_S$.

\textbf{Necessity:} The necessity is proven by contradiction. Assume
that $\mathop {||}\limits_{i = 1}^n P_i(A_S)\prec A_S$, but $DC1$,
$DC2$ or $DC3$ is not satisfied.

If $DC1$ is violated, then $\exists e_1, e_2 \in E$, $q\in Q$,
$\nexists E_i \in\{E_1,\cdots,E_n\}$, $\{e_1,e_2\}\subseteq E_i$,
$[\delta(q, e_1)!\wedge \delta(q, e_2)!]\wedge \neg [\delta(q,
e_1e_2)!\wedge \delta(q, e_2e_1)!]$. However, $\delta(q, e_1)!\wedge
\delta(q, e_2)!$, from the definition of natural projection, implies
that $\delta_i([q]_i, e_1)!\wedge \delta_j([q]_j, e_2)!$, in
$P_i(A_S)$ and $P_j(A_S)$, respectively, $\forall i\in loc(e_1),
j\in loc(e_2)$. This in turn, from definition of parallel
composition leads to $\delta_{||}(([q]_1,\cdots, [q]_n),\\ e_1)!
\wedge \delta_{||}(([q]_1,\cdots, [q]_n), e_2)!$ and
$\delta_{||}(([q]_1,\cdots, [q]_n), e_1e_2)! \wedge
\delta_{||}(([q]_1,\cdots, [q]_n), e_2e_1)!$. This means that
$\delta_{||}(([q]_1,\cdots, [q]_n), e_1e_2)! \wedge
\delta_{||}(([q]_1,\cdots, [q]_n), e_2e_1)!$ in $\mathop
{||}\limits_{i = 1}^n P_i(A_S)$, but $\neg [\delta(q, e_1e_2)!
\wedge \delta(q, e_2e_1)!]$ in $A_S$, i.e., $\mathop {||}\limits_{i
= 1}^n P_i(A_S)\nprec A_S$ which contradicts with the hypothesis.

If $DC2$ is not satisfied, then $\exists e_1, e_2 \in E$, $q\in Q$,
$\nexists E_i \in\{E_1,\cdots,E_n\}$, $\{e_1,e_2\}\subseteq E_i$,
$s\in E^*$, $ \neg [\delta(q, e_1e_2s)!\Leftrightarrow \delta(q,
e_2e_1s)!]$, i.e., $[\delta(q, e_1e_2s)!\vee \delta(q,
e_2e_1s)!]\wedge \neg [\delta(q, e_1e_2s)!\wedge \delta(q$,
$e_2e_1s)!]$. The expression $[\delta(q, e_1e_2s)!\vee \delta(q,
e_2e_1s)!]$ from definition of natural projection and Lemma
\ref{always As is similar to parallel decomposition}, respectively
implies that $\delta_{||}(([q]_1,\cdots, [q]_n), e_1e_2)! \wedge
\delta_{||}(([q]_1,\cdots, [q]_n), e_2e_1)!$ and
$\delta_{||}(([q]_1,\cdots, [q]_n), e_1e_2s)! \wedge
\delta_{||}(([q]_1,\cdots, [q]_n), e_2e_1s)!$ in $\mathop
{||}\limits_{i = 1}^n P_i(A_S)$. This in turn leads to\\
$\delta_{||}(([q]_1,\cdots, [q]_n), e_1e_2s)! \wedge
\delta_{||}(([q]_1,\cdots, [q]_n), e_2e_1s)!$ in $\mathop
{||}\limits_{i = 1}^n P_i(A_S)$, but $\neg [\delta(q, e_1e_2s)!
\wedge\\ \delta(q, e_2e_1s)!]$ in $A_S$, that contradicts with
$\mathop {||}\limits_{i = 1}^n P_i(A_S)\prec A_S$.

The violation of $DC3$ also leads to contradiction as $\delta(q_0,
s_i)!$, $i = 1,\cdots, n$, results in $\delta_{||}(([q_0]_1,\cdots
[q_0]_n), \mathop {|}\limits_{i = 1}^n \overline{p_i(s_i)})!$ in
$\mathop {||}\limits_{i = 1}^n P_i(A_S)$, whereas $\neg \delta(q_0,
\mathop {|}\limits_{i = 1}^n \overline{p_i(s_i)})!$ in $A_S$.

\section{Proof for Lemma \ref{symmetric}}\label{Appendix symmetric}
\textbf{Sufficiency:} Following two lemmas are used in the proof of
Lemma \ref{symmetric}.
\begin{lemma}(Lemma $9$ in \cite{Automatica2010-2-agents-decomposability})\label{symmetric simulations and
determinism} Consider two automata $A_1$ and $A_2$, and let $A_1$ be
deterministic, $A_1\prec A_2$ with the simulation relation $R_1$ and
$A_2\prec A_1$ with the simulation relation $R_2$. Then, $R_1^{-1} =
R_2$ if and only if there exists a deterministic automaton
$A^{\prime}_1$ such that $A_1^{\prime}\cong A_2$.
\end{lemma}
Next, let $A_1$ and $A_2$ be substituted by $A_S$ and $\mathop
{||}\limits_{i = 1}^n P_i(A_S)$, respectively, in Lemma
\ref{symmetric simulations and determinism}. Then, the existence of
$A_1^{\prime} = A_S^{\prime}$ in Lemma \ref{symmetric simulations
and determinism} is characterized by the following lemma.
\begin{lemma}\label{decomposition and determinism}
Consider a deterministic automaton $A_S$ and its natural projections
$P_i(A_S)$, $i = 1,\cdots, n$. Then, there exists a deterministic
automaton $A^{\prime}_S$ such that $A^{\prime}_S\cong \mathop
{||}\limits_{i = 1}^n P_i(A_S)$ if and only if there exist
deterministic automata $P_i^{\prime}(A_S)$ such that
$P_i^{\prime}(A_S) \cong P_i(A_S)$, $i = 1, \cdots, n$.
\end{lemma}
\begin{proof}
Let $A_S = (Q, q_0, E = \mathop {\cup}\limits_{i = 1}^n E_i,
\delta)$, $P_i(A_S) = (Q_i, q_0^i, E_i, \delta_i)$,
$P_i^{\prime}(A_S) = (Q_i^{\prime}, q_{0,i}^{\prime}, E_i,
\delta_i^{\prime})$, $i = 1, \cdots, n$, $\mathop {||}\limits_{i =
1}^n P_i(A_S) = (Z, z_0, E, \delta_{||})$, $\mathop {||}\limits_{i =
1}^n P_i^{\prime}(A_S) = (Z^{\prime}, z_0^{\prime}, E,
\delta_{||}^{\prime})$. Then, the proof of Lemma \ref{decomposition
and determinism} is presented as follows.

\textbf{Sufficiency:} The existence of deterministic automata
$P_i^{\prime}(A_S)$ such that $P_i^{\prime}(A_S) \cong P_i(A_S)$, $i
= 1, \cdots, n$ implies that $\delta^{\prime}_i$, $i = 1, \cdots, n$
are functions, and consequently from definition of parallel
composition (Definition \ref{parallel composition}),
$\delta_{||}^{\prime}$ is a function, and hence $\mathop
{||}\limits_{i = 1}^n P_i^{\prime}(A_S)$ is deterministic. Moreover,
from Lemma \ref{ParallelBiSimulation}, $P_i^{\prime}(A_S) \cong
P_i(A_S)$, $i = 1,\cdots, n$ lead to $\mathop {||}\limits_{i = 1}^n
P_i^{\prime}(A_S)\cong \mathop {||}\limits_{i = 1}^n P_i(A_S)$,
meaning that there exists a deterministic automaton $A^{\prime}_S :=
\mathop {||}\limits_{i = 1}^n P_i^{\prime}(A_S)$ such that
$A^{\prime}_S \cong \mathop {||}\limits_{i = 1}^n P_i(A_S)$.

\textbf{Necessity:} The necessity is proven by contraposition,
namely, by showing that if there does not exist deterministic
automata $P_i^{\prime}(A_S)$ such that $P_i^{\prime}(A_S) \cong
P_i(A_S)$, for $i = 1, 2, \cdots, \hbox{ or } n$, then there does
not exist a deterministic automaton $A^{\prime}_S$ such that
$A^{\prime}_S \cong \mathop {||}\limits_{i = 1}^n P_i(A_S)$.

Without loss of generality, assume that there does not exist a
deterministic automaton $P_1^{\prime}(A_S)$ such that
$P_1^{\prime}(A_S) \cong P_1(A_S)$. This means that $\exists q, q_1,
q_2 \in Q$, $e\in E_1$, $t_1, t_2\in (E\backslash E_1)^*$, $t\in
E^*$, $\delta(q, t_1e) = q_1$, $\delta(q, t_2e) = q_2$,
$\neg[\delta(q_1, t)!\Leftrightarrow \delta(q_2, t)!]$, meaning that
$\delta(q_1, t)!\wedge \neg\delta(q_2, t)!$ or $\neg\delta(q_1,
t)!\wedge \delta(q_2, t)!$. Again without loss of generality we
consider the first case and show that it leads to a contradiction.
The contradiction of the second case is followed, similarly. From
the first case, $\delta(q_1, t)!\wedge \neg\delta(q_2, t)!$,
definition of natural projection, definitions of parallel
composition and Lemma \ref{always As is similar to parallel
decomposition}
 it follows that
 $([q_1]_1, ([q_1]_2, \ldots, [q_1]_n))\in \delta_{||}(([q]_1, ([q]_2, \ldots, [q]_n)), t_1e)$,
 $([q_2]_1, ([q_1]_2, \ldots, [q_1]_n))\in \delta_{||}(([q]_1, ([q]_2, \ldots, [q]_n)), t_1e)$,
 $\delta(([q_1]_1, ([q_1]_2, \ldots, [q_1]_n)), t)!$, whereas  $\neg\delta(([q_2]_1, ([q_1]_2, \ldots, [q_1]_n)), t)!$
in $\mathop {||}\limits_{j = 1}^n P_i(A_S)$, implying that there
does not exist a deterministic automaton $A^{\prime}_S$ such that
$A^{\prime}_S \cong \mathop {||}\limits_{j = 1}^n P_i(A_S)$, and the
necessity is followed.
\end{proof}

Now, Lemma \ref{symmetric} is proven as follows.

\textbf{Sufficiency:} $DC4$ implies that there exist deterministic
automata $P_i^{\prime}(A_S)$ such that $P_i^{\prime}(A_S) \cong
P_i(A_S)$, $i = 1, \cdots, n$. Then, from Lemmas
\ref{ParallelBiSimulation} and \ref{decomposition and determinism},
it follows, respectively, that $\mathop {||}\limits_{i = 1}^n
P_i^{\prime}(A_S)\cong \mathop {||}\limits_{i = 1}^n P_i(A_S)$, and
that there exists a deterministic automaton $A^{\prime}_S := \mathop
{||}\limits_{i = 1}^n P_i^{\prime}(A_S)$ such that $A^{\prime}_S
\cong \mathop {||}\limits_{i = 1}^n P_i(A_S)$ that due to Lemma
\ref{symmetric simulations and determinism}, it results in $R_1^{-1}
= R_2$.

\textbf{Necessity:} Let $A_S$ be deterministic, $A_S\prec \mathop
{||}\limits_{i = 1}^n P_i(A_S)$ with the simulation relation $R_1$
and $\mathop {||}\limits_{i = 1}^n P_i(A_S)\prec A_S$ with the
simulation relation $R_2$, and assume by contradiction that
$R_1^{-1} = R_2$, but $DC4$ is not satisfied. Violation of $DC4$
implies that for $\exists i \in \{1,\cdots, n\}$, there does not
exists a deterministic automaton $P_i^{\prime}(A_S)$ such that
$P_i^{\prime}(A_S) \cong P_i(A_S)$. Therefore, due to Lemma
\ref{decomposition and determinism}, there does not exist a
deterministic automaton $A^{\prime}_S$ such that $A^{\prime}_S \cong
\mathop {||}\limits_{i = 1}^n P_i(A_S)$, and hence, according to
Lemma \ref{symmetric simulations and determinism}, it leads to
$R_1^{-1} \neq R_2$ which is a contradiction.

\bibliographystyle{IEEEtran}
\bibliography{Ref_13}

\end{document}